\documentclass[reprint,aps,prl,showpacs,floatfix,superscriptaddress,nofootinbib]{revtex4-1}

\usepackage{graphicx}
\usepackage{amsthm}   
 \usepackage[bottom]{footmisc}
\usepackage{amsmath} 

\usepackage{amssymb}  
\usepackage{mathrsfs} 
\usepackage{stmaryrd} 
\usepackage{txfonts} 

\usepackage{enumerate}
\usepackage{appendix} 

\newcommand{\mathsym}[1]{{}}
\newcommand{\unicode}[1]{{}}

\newcommand{\R}{{\mathbb{R}}}
\newcommand{\B}{{\mathbb{B}}}

\newcommand{\DDD}{{\mathcal{D}}}

\usepackage{color}  
\RequirePackage[dvipsnames,usenames]{xcolor}

\newtheorem{theorem}{Theorem}

\newtheorem{corollary}[theorem]{Corollary}

\newtheorem{definition}{Definition}
\newtheorem{example}{Example}

\newtheorem{lemma}[theorem]{Lemma}

\newtheorem{proposition}[theorem]{Proposition}

\newcommand{\be}{\begin{equation}}
\newcommand{\bel}[1]{\begin{equation}\label{#1}}
\newcommand{\qe}{\end{equation}}
\newcommand{\ee}{\end{equation}}
\newcommand{\eeq}{\end{equation}}
\newcommand{\ba}{\begin{eqnarray}}
\newcommand{\ea}{\end{eqnarray}}

\begin{document}

\title{Constraints on physical reality arising from a formalization of knowledge}

\author{David H. Wolpert \\
 Santa Fe Institute, 1399 Hyde Park Road, Santa Fe, NM, 87501\\
 \texttt{david.h.wolpert@gmail.com}
 }


\begin{abstract}
There are (at least) four ways that an agent can acquire information concerning the state of the universe: via
observation,  control,  prediction, and via retrodiction, i.e., memory. Each of these four ways of acquiring information seems
to rely on a different kind of physical device (resp., an observation device, a control device, etc.).
However it turns out that certain mathematical structure is common to those four types of 
device. Any device that possesses a certain subset of that structure is known as an ``inference device'' (ID). 

Here I review some of the properties of IDs, including their relation with Turing
machines, and (more loosely) quantum mechanics. I also review the bounds
of the joint abilities of any set of IDs to know facts about the physical
universe that contains them. These bounds constrain the possible 
properties of any universe that contains agents who can acquire information
concerning that universe.

I then extend this previous work on IDs, by adding to the definition of IDs some of the other mathematical structure
that is common to the four ways of acquiring information about the universe but is not captured in the 
(minimal) definition of IDs. I discuss these extensions of IDs in the context of 
epistemic logic (especially possible worlds formalisms like Kripke structures and Aumann
structures). In
particular, I show that these extensions of IDs are not subject to the problem of logical omniscience that plagues 
many previously studied forms of epistemic logic.

\end{abstract}

\maketitle

\section{Introduction}

Ever since Wheeler discussed ``it from bit"~\cite{whee90}, there has been great
interest in what constraints on the properties of the universe can be derived using some appropriate
mathematical formulation of information. Some of this work relies on Shannon information theory~\cite{caticha2011entropic,goyal2010information,goyal2012information,hartle2005physics},
and some of it on Fisher information theory~\cite{frieden2004science}. There has also been 
work on this topic that focuses on the \emph{processing} of information, i.e., that views the universe through
the lens of Turing machine (TM) theory~\cite{lloyd1990valuable,chaitin2004algorithmic,
zure89a,zure89b,zurek1990complexity,zenil2012computable,lloyd2006programming,schmidhuber2000algorithmic,zuse1969rechnender}.

Here I adopt a different approach. I focus on the fact that information concerning the state of the
universe typically \emph{is held by some agent embedded in that universe}. 
For example, we cannot speak of Shannon information without
specifying probability distributions --- which reflect the uncertainty of some specific agent concerning
the state of the universe that contains them
(e.g., uncertainty of a scientist making a prediction). 

This leads us to analyze how an agent can acquire information
concerning the state of the world. That is the topic of this paper,
as described in the remainder of this introduction.

\subsection{Inference devices}


There are (at least) four ways
an agent can acquire some information concerning the universe in which it is 
embedded: via an observation
device, via a control device, via a prediction device, and via a memory device, i.e., a ``retrodiction'' device.
It turns out that there is some mathematical structure shared by 
all such information-acquiring devices. Devices with that structure are called
``Inference Devices" (IDs)~\cite{wolp01,wolp08b,wolp10,bind08}. 

In the first section of this paper I present two examples of how an agent can acquire
information about the universe that contains them, illustrating that in both examples
the agent has the mathematical structure of an ID.
I then present some of the more elementary impossibility results concerning IDs.
These results place strong constraints on what information about a universe can be
jointly held by different IDs embedded in that universe. Importantly, these
constraints arise only from the definition of IDs, without any assumptions about the laws
of the universe containing the IDs; they hold in \emph{any} universe that allows
agents that have information about that universe. 
In particular, they would hold even in a classical, finite universe, with
no chaotic processes. They would also hold even in a universe with agents
who have super-Turing computation abilities, can transmit information at super-luminal rates, etc.
It is worth noting as well that these impossibility theorems hold even though there is
no sense in which IDs have the ability of self-reference.

After these preliminaries I present some of the connections between the theory of IDs and the 
theory of Turing Machines. In particular I analyze some of the properties
of an ID version of universal Turing machines and of an ID version of Kolmogorov 
complexity~\cite{lloyd1990valuable,chaitin2004algorithmic,
zure89a,zure89b,zurek1990complexity,zenil2012computable,lloyd2006programming,schmidhuber2000algorithmic,zuse1969rechnender}.
I show that the ID versions of those quantities obey many of the familiar results of Turing machine theory (e.g., 
the invariance theorem of TM theory).

I then consider one way to extend the theory of IDs
to the case where there is a probability distribution over the states of the universe,
so that no information is ever 100\% guaranteed to be true. In particular,
I present a result concerning the products of probabilities of error of two separate IDs, a result
which is formally similar to the Heisenberg uncertainty principle. 

These results all concern subsets of an entire universe, e.g., one or two IDs embedded in 
a larger universe. However we can expand the scope to an entire
universe. The idea is to \emph{define} a ``universe'', with whatever associated laws of physics, to be a set of 
physical systems and IDs (e.g., a set of scientists), where the IDs can have
information concerning those physical systems and / or one another. 
Adopting this approach, I use the theory of IDs to derive impossibility results concerning 
the nature of the entire universe.

\subsection{Inference devices and epistemic logic}

Most of the results presented to this point in the paper have appeared before,
albeit in a more complicated, less transparent formalism than the one used here~\cite{wolp01,wolp08b,wolp10,bind08}.
In the last sections of the paper I present new results. These
all involve an extension of IDs, one that
includes some of the features that are shared by the four ways for an agent to acquire information
(observation, control, prediction and memory) but that are not in the original definition of IDs. 

I show that this strengthened version of IDs has a close relation to the various ways of
formalizing ``knowledge'' that are considered in epistemic logic~\cite{fagin2004reasoning,aaronson2013philosophers,parikh1987knowledge,zalta2003stanford}.
However the ID-based theory of knowledge is not subject to what is perhaps the major problem
of these earlier ways of formalizing knowledge, the problem of \emph{logical omniscience}.

To explain that problem, it is easiest to work with the \emph{event-based}
formulation of epistemic logic pioneered by Aumann~\cite{aumann1976agreeing,auma99,aubr95,futi91,bibr88}.
In this formulation we start with a space $U$ of possible states of the entire universe
across all time. (In the literature this is typically called a set of ``possible worlds''.)
An \emph{event} is defined as a subset of $U$. For example, let $U$ be a set of all
possible histories of the universe across all time and space. Then the event  
\{there are no clouds in the sky in London during January 1, 2000\} is all 
$u \in U$ such that \{there are no clouds in the sky in London during January 1, 2000\}. Belief
by an agent concerning the universe is formalized in terms of a  partition
of $U$, which is called an \emph{Aumann structure} and
an associated \emph{belief operator} $B : U \rightarrow \{R_i\}$.
This is supposed to represent the intuition that
in any universe $u$, the agent \emph{believes} that an event $E \subset U$ holds
iff $B(u) \subseteq E$. \emph{Knowledge} is then defined as true belief, i.e., a belief operator $K$ with
the property that $u \in K(u)$ for all $u \in U$~\cite{fagin2004reasoning}.
Similarly, in the event-based framework we say 
that ``an agent\textit{knows event $E$} in world $u \in U$'' if $E$ holds 
for all worlds that the agent believes are possible when the actual world is $u$. 
So an agent knows $E$ at $u$ if for their knowledge operator, $K(u) \subseteq E$.

Now suppose that some event $E$ implies some event $E'$, i.e., $E' \supseteq E$. 
This means that under the event-based definition of knowledge, if agent $i$ knows 
event $E$ in world $u$, $E'$ is also true in world $u$ --- and agent $i$ knows event $E'$ in world $u$.
Generalizing, the agent cannot know a set of facts without knowing 
all logical implications of those facts.
This is known as the property of \textbf{logical omniscience}. As an example of this
property, suppose that someone
multiplies two huge prime numbers and then (honestly) tells that product to the agent --- so that the agent
knows that product. Then since that product logically implies its unique
prime factorization, under the event-based framework 
the agent must ``know'' the two prime numbers, independent
of any considerations of whether they have a computer to help them do calculations.
This of course is absurd.
%

This problem of logical omniscience plagues possible-worlds models
of epistemic logic like those based on Aumann structures. Some extensions to possible-worlds
models have been proposed to address this problem, e.g., bounds 
on the computational powers of the agent~\cite{aaronson2013philosophers}, assuming
that the agent reasons illogically~\cite{fagin2004reasoning}, and a
set of related ``impossible possible worlds'' restrictions on the nature of the agent~\cite{fagin2004reasoning,aaronson2013philosophers,parikh1987knowledge,zalta2003stanford}. 
However none of these has proven broadly convincing.

There are other difficulties with the event-based formalization of knowledge.
In that framework, by simply defining the ``knowledge operator'' of a rock on the moon appropriately,
we would say that the rock ``knows'' whether it is in sunlight or not (for example by having its knowledge
operator pick sets of states of the universe based on the temperature of that rock). This pathology is due to the
fact that the definition of knowledge operators in the event-based framework
does not reflect the fact that knowledge is \emph{held by
a sentient agent}. Specifically, any sentient agent that knows something about the universe
is able to correctly answer arbitrary questions
about what they know, either implicitly or explicitly. (Note that a lunar rock cannot answer such questions.) However there
is nothing in the formal structure of Aumann structures, Kripke structures, 
or the like that involves the ability of agents to correctly answer questions.

The ID framework is concerned precisely with
such ability of an agent to answer questions about the information they have. As a result, in
the extension of the ID framework into a full-fledged theory of knowledge,
we cannot say that a rock on the moon ``knows'' whether it is in sunlight.
Moreover, the ID-based theory of knowledge avoids the problem
of logical omniscience. Specifically, in the ID-based formalization of knowledge, if an agent knows some fact $A$, and knows
that $A$ implies $B$, then $B$ is true --- but the agent may not know that. 

None of the results below are difficult to prove; some of the proofs, especially
those of the ``Laplace demon theorems'', are almost trivial. (The interest is the 
implications of the inference device axioms for metaphysics and epistemology, not the math needed
to derive those implications.) Nonetheless, the interested reader can find
all proofs that are not given below in~\cite{wolp08b}.

\section{Inference Devices}
\label{sec:review}

In this section I review the elementary properties of inference devices, mathematical structures that
are shared by the processes of observation, prediction, recall and control~\cite{wolp01,wolp08b,wolp10,bind08}.
These results are proven by extending Epimenides' paradox to apply
to novel scenarios. Results relying on more sophisticated mathematics, 
some of them new, are presented the following section.

%
%

\subsection{Observation, prediction, recall and control of the physical world}
\label{sec:physical}

I begin with two examples that motivate the formal definition
of inference devices. The first is
an example of an agent making a correct observation
about the current state of some physical variable. 

\begin{example}
Consider an agent who claims to be able to observe $S(t_2)$, the
value of some physical variable at time $t_2$. If the agent's
claim is correct, then for any question of the form ``Does $S(t_2) =
L$?", the agent is able to consider that question at some $t_1 <
t_2$, observe $S(t_2)$, and then at some $t_3 > t_2$ provide the
answer ``yes" if $S(t_2) = L$, and the answer ``no" otherwise.  
In other words, she can correctly pose any such binary question to herself
at $t_1$, and correctly say what the answer is at
$t_3$.{\footnote{It may be that the agent has to use some
appropriate observation apparatus to do this; in that case we can just
expand the definition of the ``agent" to include that apparatus.
Similarly, it may be that the agent has to configure that
apparatus appropriately at $t_1$. In this case, just expand our
definition of the agent's ``considering the appropriate question"
to mean configuring the apparatus appropriately, in addition to the
cognitive event of her considering that question.}}

To formalize this,
let $U$ refer to a set of possible histories of an entire universe across all time, 
where each $u \in U$ has the following properties:
\begin{enumerate}[i)]
\item $u$ is consistent with the laws of physics,
\item In $u$, the agent is alive and of sound mind throughout the time interval $[t_1, t_3]$, and the system $S$ exists
at the time $t_2$,
\item In $u$, at time $t_1$ the agent considers some  $L$-indexed  question $q$
of the form ``Does $S(t_2) = L$?",
\item In $u$, the agent observes $S(t_2)$, 
\item In $u$, at time $t_3$ the agent uses that observation to
provide her (binary) answer to $q$, and believes that answer to be 
correct.{\footnote{This means in particular that the agent does not lie, does not believe she was distracted
from the question during $[t_1, t_3]$.}}
\end{enumerate}
The agent's claim is that for any question $q$ of the form ``Does
$S(t_2) = L$?", the laws of physics imply that for all $u$ in the
subset of $U$ where at $t_1$ the agent considers $q$, it
must be that the agent provides the correct answer to $q$ at
$t_3$.  Any prior knowledge concerning the history that the
agent relies on to make this claim is embodied in the set $U$.

The value $S(t_2)$ is a function of the actual history of the entire
universe, $u \in U$.  Write that function as $\Gamma(u)$, with
image $\Gamma(U)$.  Similarly, the question the agent has in her brain at
$t_1$, together with the time $t_1$ state of any observation apparatus she
will use, is a function of $u$. Write that function as $X(u)$.
Finally, the binary answer the agent provides at $t_3$ is a function of the
state of her brain at $t_3$, and therefore it too is a function of
$u$. Write that binary-valued function giving her answer as $Y(u)$. 

Note that since
$U$ embodies the laws of physics, in particular it embodies
all neurological processes in the agent (e.g., her asking and
answering questions), all physical characteristics of $S$, etc.

So as far as this observation is concerned, the agent is just a
pair of functions $(X, Y)$, both with the domain $U$ defined above,
where $Y$ has the range $ \{-1, 1\}$. A necessary condition for us to say that the agent can
``observe $S(t_2)$" is that for any
$\gamma \in \Gamma(U)$, there is some associated $X$ value $x$ such that
for all $u \in U$, so long as $X(u) = x$, it follows that  $Y(u) = 1$ iff
$\Gamma(u) = \gamma$.
\label{ex:beg_1}
\end{example}

I now present an example of an agent making a correct prediction
about the future state of some physical variable. 
\begin{example}
Now consider an agent who claims to be able to predict $S(t_3)$, the
value of some physical variable at time $t_3$. If the agent's
claim is correct, then for any question of the form ``Does $S(t_3) =
L$?", the agent is able to consider that question at some time $t_1 <
t_3$, and produce an answer at some time $t_2 \in (t_1, t_3)$, where the answer is ``yes"
if  $S(t_3) = L$ and ``no" otherwise.  So
loosely speaking, if the agent's claim is correct, then for any $L$, by
their considering the appropriate question at $t_1$, they can generate the correct
answer to any question of the form ``Does $S(t_3) = L$?" at
$t_2 < t_3$.{\footnote{\label{foot:4} It may be that the agent has to use some
appropriate prediction computer to do this; in that case we can just
expand the definition of the ``agent" to include that computer.
Similarly, it may be that the agent has to program that
computer appropriately at $t_1$. In this case, just expand our
definition of the agent's ``considering the appropriate question"
to mean programming the computer appropriately, in addition to the
cognitive event of his considering that question.}}

To formalize this, let $U$ refer to a set of possible histories of an entire universe across all time, 
where each $u \in U$ has the following properties:
\begin{enumerate}[i)]
\item $u$ is consistent with the laws of physics,
\item In $u$, the agent exists throughout the interval $[t_1, t_2]$, and the system $S$ exists
at $t_3$,
\item In $u$, at $t_1$ the agent considers some question $q$
of the form ``Does $S(t_3) = L$?",
\item In $u$, at $t_2$ the agent provides his (binary) answer to $q$ 
and believes that answer to be 
correct.{\footnote{This means in particular that the agent does not believe he was distracted
from the question during $[t_1, t_2]$.}}
\end{enumerate}
The agent's claim is that for any question $q$ of the form ``Does
$S(t_3) = L$?", the laws of physics imply that for all $u$ in
the restricted set $U$ such that
at $t_1$ the agent considers $q$, it must be
that the agent provides the correct answer to $q$ at
$t_2$.

The value $S(t_3)$ is a function of the actual history of the entire
universe, $u \in U$.  Write that function as $\Gamma(u)$, with
image $\Gamma(U)$.  Similarly, the question the agent considers at
$t_1$ is a function of the state of his brain at $t_1$, and therefore
is also a function of $u$. Write that function as $X(u)$. Finally,
the binary answer the agent provides at $t_2$ is a function of the
state of his brain at $t_2$, and therefore it too is a function of of
$u$. Write that function as $Y(u)$. 

So as far as this prediction is concerned, the agent is just a
pair of functions $(X, Y)$, both with the domain $U$ defined above,
where $Y$ has the range $ \{-1, 1\}$. The agent can
indeed predict $S(t_3)$ if for the space defined above $U$, for any
$\gamma \in \Gamma(U)$, there is some associated $X$ value $x$ such that, no
matter what precise history $u \in U$ we are in, due to the laws of
physics, if $X(u) = x$ then the associated $Y(u)$ equals $1$ iff
$\Gamma(u) = \gamma$.
\label{ex:beg_2}
\end{example}

Evidently, agents who perform observation and
those who perform prediction are described in part by a shared mathematical structure, involving functions
 $X$ and $Y$ defined over the same space $U$ of all possible histories of the universe
across all time.   As formalized below, I refer to
any such pair $(X, Y)$ as an ``inference device". 
Say that for some function $\Gamma$ defined over $U$, for any 
$\gamma \in \Gamma(U)$, there is some associated $X$ value $x$ such that, no
matter what precise history $u \in U$ we are in, due to the laws of
physics, if $X(u) = x$ then the associated $Y(u)$ equals $1$ iff
$\Gamma(u) = \gamma$. Then I will say that the device $(X, Y)$
``infers" $\Gamma$.

See~\cite{wolp08b} for a more detailed elaboration of the examples
given above
of observation and prediction in terms of inference devices. 
Arguably to fully formalize each of these
phenomena  there should be additional structure beyond that
defining inference devices.  (See App. B. of~\cite{wolp08b}.) 
Most such additional structure is left for future research. However,
one particular part of such additional structure is investigated below,
in the discussion of ``physical knowledge''.

In addition to considering observation and prediction,
it is also shown in \cite{wolp08b} that a system that remembers the past 
is an inference device that infers an appropriate function $\Gamma(u)$.\footnote{Loosely
speaking, memory is just retrodiction, i.e., it is using current data to
predict the state of non-current data. However, rather than have the
non-current data concern the future, in memory it concerns the past.}
\cite{wolp08b} also shows that a device that controls a physical variable is an
inference device that infers an appropriate function $\Gamma(u)$. All of this 
analysis holds even if what is observed / predicted / remembered / controlled
is not the answer to a question of the form, ``Does $S(t) = L$?", but
instead an answer to  question of the form, ``is $S(t)$ more property $A$ than it is property $B$?"
or of the form, ``is $S(t)$ more property $A$ than $S'(t)$ is?"
%

In the sequel I will sometimes consider situations involving multiple
inference devices, $(X_1, Y_1), (X_2, Y_2), \ldots$, with associated
domains $U_1, U_2, \ldots$. For example, I will consider scenarios
where agents try to observe one another. In such situations, when
referring to ``$U$'', I implicitly mean $\cap_i U_i$, implicitly
restrict the domain of all $X_i, Y_i$ to $U$, and implicitly assume
that the codomain of each such restricted $Y_i$ is binary.

\subsection{Notation and terminology}

To formalize the preceding considerations, I first fix some notation.
I will take the set of binary numbers $\mathbb{B}$ to equal $\{-1,
1\}$. In the canonical case where $U$ is the set of all possible histories of the entire universe
across all space and time, the value of
any specific physical variable is specified by a subset of the components of
a full vector $u \in U$. (For example, the variable of the speed of a particular
particle in a particular intertial frame at a particular time is given
by a subset of the components of $u$.) So any such variable is just 
a function over $U$. Bearing this in mind, for any
function $\Gamma$ with domain $U$, I will write the image of $U$ under
$\Gamma$ as $\Gamma(U)$, i.e., for the set of possible values of some
physical variable. I will also sometimes abuse this notation with a sort of
``set-valued function'' shorthand, and so for example write $\Gamma(V) = 1$ for some $V \subset U$
iff $\Gamma(u) = 1 \; \forall u \in V$. On the other hand, for the special case where the function over $U$
is a measure, I use conventional shorthand from measure theory. For example,
if $P$ is a probability distribution over $U$ and $V \subset U$, I write $P(V)$ as shorthand for $\sum_{u \in V} P(u)$.

For any function $\Gamma$ with domain $U$
that I will consider, I implicitly assume that the entire set $\Gamma(U)$ contains at
least two distinct elements. For any (potentially infinite) set $R$,
$|R|$ is the cardinality of $R$. 

Given a function $\Gamma$
with domain $U$, I write the partition of $U$ given by  $\Gamma^{-1}$ as $\overline{\Gamma}$,
i.e.,
\begin{eqnarray}
{\overline{\Gamma}} &\equiv& \{ \{u : \Gamma(u) = \gamma\} : \gamma \in \Gamma(U)\}
\end{eqnarray}
I say that two functions $\Gamma_1$
and $\Gamma_2$ with the same domain $U$ are \textbf{(functionally) equivalent} iff
the inverse functions $\Gamma_1^{-1}$ and $\Gamma_2^{-1}$ induce the same partitions
of $U$, i.e., iff ${\overline{\Gamma_1}} = {\overline{\Gamma_2}}$.

Recall that a partition
$A$ over a space $U$ is a {\emph{refinement}} of a partition $B$ over
$U$  iff every $a \in A$ is a subset of some $b \in B$. If $A$ is a refinement
of $B$, then for every $b \in B$ there is an $a \in A$ that is a subset of $b$. 
Some of the elementary properties of refinement will be used below, and
so I now review them. First, two partitions $A$ and $B$
are refinements of each other iff $A = B$. Say a partition $A$ is
finite and a refinement of a partition $B$. Then $|A| = |B|$ iff $A
= B$. For any two functions $A$ and $B$ with domain $U$, I will say that ``$A$ refines $B$" if ${\overline{A}}$ is a refinement
of $\overline{B}$. Similarly, for any $R \subset U$ and function $A$, I will say that ``$R$ refines
$A$" (or ``$A$ is refined by $R$") if $R$ is a subset of some element of $\overline{A}$.


I write the characteristic function of any set $R \subseteq U$ as the binary-valued function
\begin{eqnarray}
\mathcal{X}_R(u) = 1 \Leftrightarrow u \in R 
\label{eq:shorthand1}
\end{eqnarray}
As shorthand I will sometimes treat
functions as equivalent to one of the values in their image. So for example
expressions like ``$\Gamma_1 = \Gamma_2 \Rightarrow \Gamma_3 = 1$''
means `$`\forall u \in U$ such that $\Gamma_1(u) = \Gamma_2(u)$, $\Gamma_3(u) = 1$''.

To simplify terminology, rather than referring to Kronecker delta functions 
(and / or Dirac delta functions) throughout, I will refer to a {\bf{probe}} of a variable $V$,
by which I mean any function over $U$ 
parametrized by a $v \in V$ of the form
\begin{equation} 
\delta_v(v') = 
\begin{cases}  
1 & \text{ if $v = v'$} \\ 
-1 & \text{ otherwise.} 
\end{cases} 
\end{equation}
$\forall v' \in V$.
%
Given a function $\Gamma$ with domain
$U$ I sometimes write 
$\delta_\gamma(\Gamma)$ as shorthand
for the function $u \in U \rightarrow \delta_\gamma(\Gamma(u))$. 
When I don't want to specify the subscript $\gamma$ of a probe, I sometimes generically write $\delta$. I write  ${\mathcal{P}}(\Gamma)$ to
indicate the set of all probes over $\Gamma(U)$.

\subsection{Weak inference}
\label{sec:weak}

I now review some results that place severe restrictions on what a physical
agent can predict / observe / control / remember
and be guaranteed to be correct (in that prediction / observation / control / memory). 
To begin, I formalize the concept of an ``inference device" 
introduced in the previous subsection.

\begin{definition}
\label{def:id}

An {\bf{(inference) device}} over a set $U$ is a pair of functions
$(X, Y)$, both with domain $U$.  $Y$ is called the {\bf{conclusion}}
function of the device, and is surjective onto $\mathbb{B}$.  $X$ is
called the {\bf{setup}} function of the device.
\end{definition}

Given some function $\Gamma$ with domain $U$ and some $\gamma \in
\Gamma(U)$, we are interested in setting up a device so that it is
assured of correctly answering whether $\Gamma(u) = \gamma$ for the
actual universe $u$.  Motivated by the examples above, I will formalize this with
the condition that $Y(u) = 1$ iff $\Gamma(u) =
\gamma$ for all $u$ that are consistent with some associated setup
value $x$ of the device, i.e., such that $X(u) = x$ for some $x$.  If this condition
holds, then setting up the device to have setup value $x$ guarantees
that the device will make the correct conclusion concerning whether
$\Gamma(u) = \gamma$. (Hence the terms ``setup function'' and
``conclusion function'' in Def. 1.)

We can formalize this as follows:
 
\begin{definition}
 \label{def:weak_inf}
 Let $\Gamma$ be a function over $U$ such that $|\Gamma(U)| \ge 2$. 
 A device $\DDD$ {\bf{(weakly) infers}}  $\Gamma$ 
 iff $\forall \gamma \in \Gamma(U)$, $\exists x \in X(U)$ such that
 $\forall u \in U, X(u) = x \Rightarrow Y(u) = \delta_\gamma({\Gamma(u)})$.
\end{definition}
 
\noindent If $\DDD$ infers $\Gamma$, I write $\DDD > \Gamma$.  I say that a device
$\DDD$ infers a set of functions if it infers every function in that set.

The following semi-formal example illustrates a scenario in which weak inference
holds, and a related scenario in which it doesn't hold.

\begin{example}
A scenario in which weak inference holds is illustrated in Fig.~\ref{fig:correct_example}. In this example,
for simplicity determinism
is assumed. The full rectangle, including both colored rectangles, indicates 
the set of all possible histories of the universe across all time, $U$ (i.e., the set of all
``states of the world'', in the language of epistemology). In this example the function $\Gamma$ is whether the sky will 
(not) be cloudy at noon (at Greenwich, say). Since the ID is
embedded in the universe, the precise question concerning the future state of the universe that it is instructed to answer
picks out different subsets of the set of all possible histories of the universe across all time. 
There are two such sets indicated, corresponding to
the ID being asked the question, ``will the sky be cloudy at noon?'' or being asked the question, ``will the sky
be clear at noon?''. (Histories falling outside of both of those sets correspond to questions different from
those two.) Again, since the ID is embedded in the universe, and since its answer can have two
possible values, which answer it gives (say at 11am) is a partition across $U$. The separatrix between the
two elements of that partition are indicated by the bold line. Finally, in all elements of $U$, the sky either
will be clear at noon or will be cloudy. The two possibilities are indicated by the two colored rectangles.

The ID weakly infers $\Gamma$, i.e., correctly predicts the state of the sky at noon, since whichever
of the two possible questions it considers, it is guaranteed that its answer is correct. 

A related scenario where weak inference does not hold is illustrated in Fig.~\ref{fig:incorrect_example}.
The only difference from the scenario depicted in Fig.~\ref{fig:correct_example} is that if the ID
is asked the question, ``will the sky be cloudy at noon?'',
and the sky in fact will be cloudy at noon, the ID will answer 'no' --- which is incorrect.
 \begin{figure}
	\hglue-5mm
  	\includegraphics[width= 1.2\linewidth]{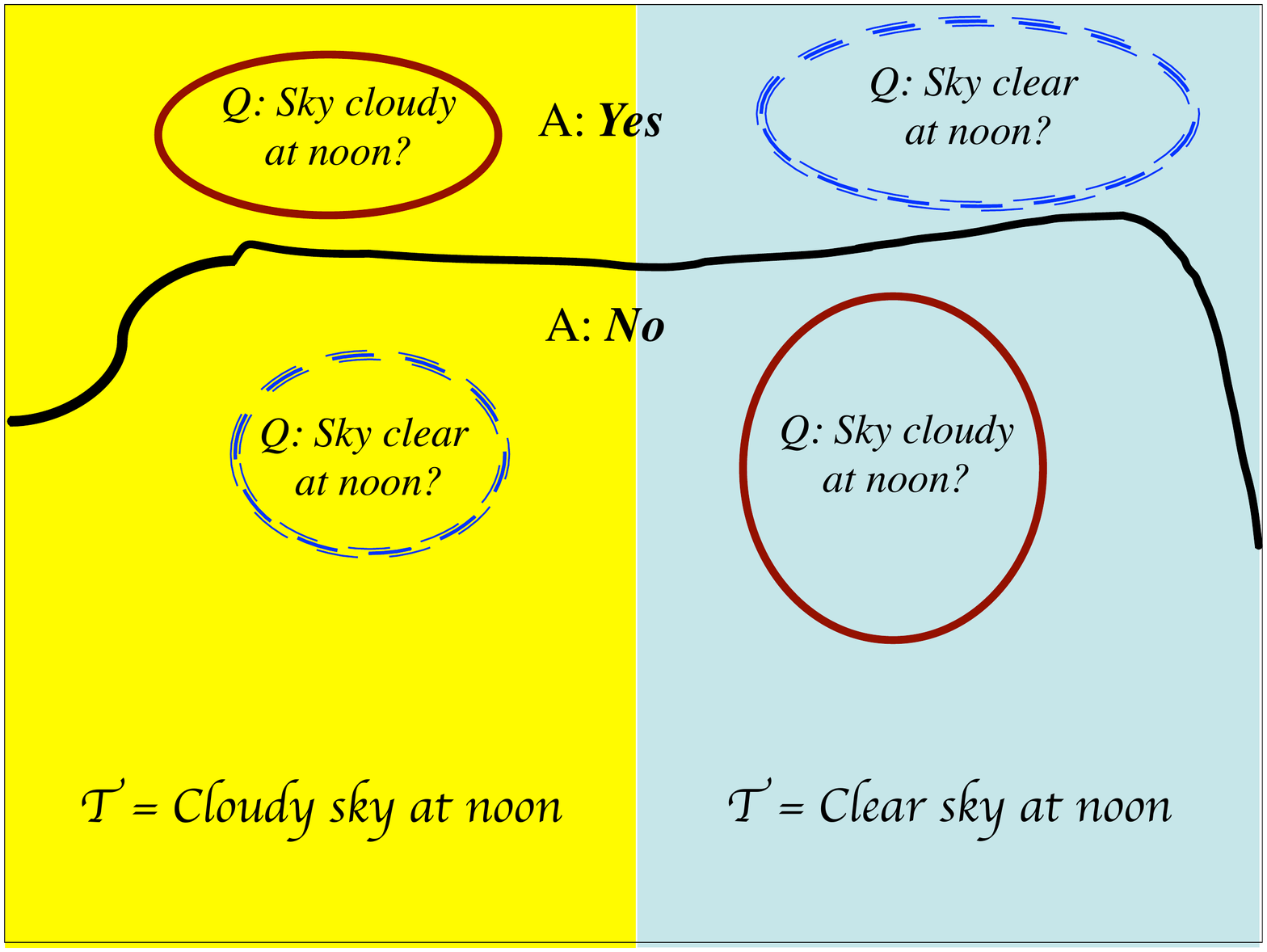}
 	\noindent \caption{An example of correct prediction as weak inference, where for simplicity determinism
is assumed. The set $U$ of all possible histories of the universe is the full rectangle, including both the yellow and blue
subsets, which correspond to the two possible states of the sky at noon.
Two of the possible questions of the ID are indicated: one of them is asked by the ID in all universes
	within the union of the two red ellipses, and the other question is asked in all universes within
	the union of the two blue ellipses. The ID weakly infers $\Gamma$, i.e., correctly predicts the state of the sky at noon, since whichever
of the two possible questions it considers, it is guaranteed that its answer is correct.
	  }
 	\label{fig:correct_example}
 \end{figure}
 
  \begin{figure}
	\hglue-10mm 
  	\includegraphics[width= 1.2\linewidth]{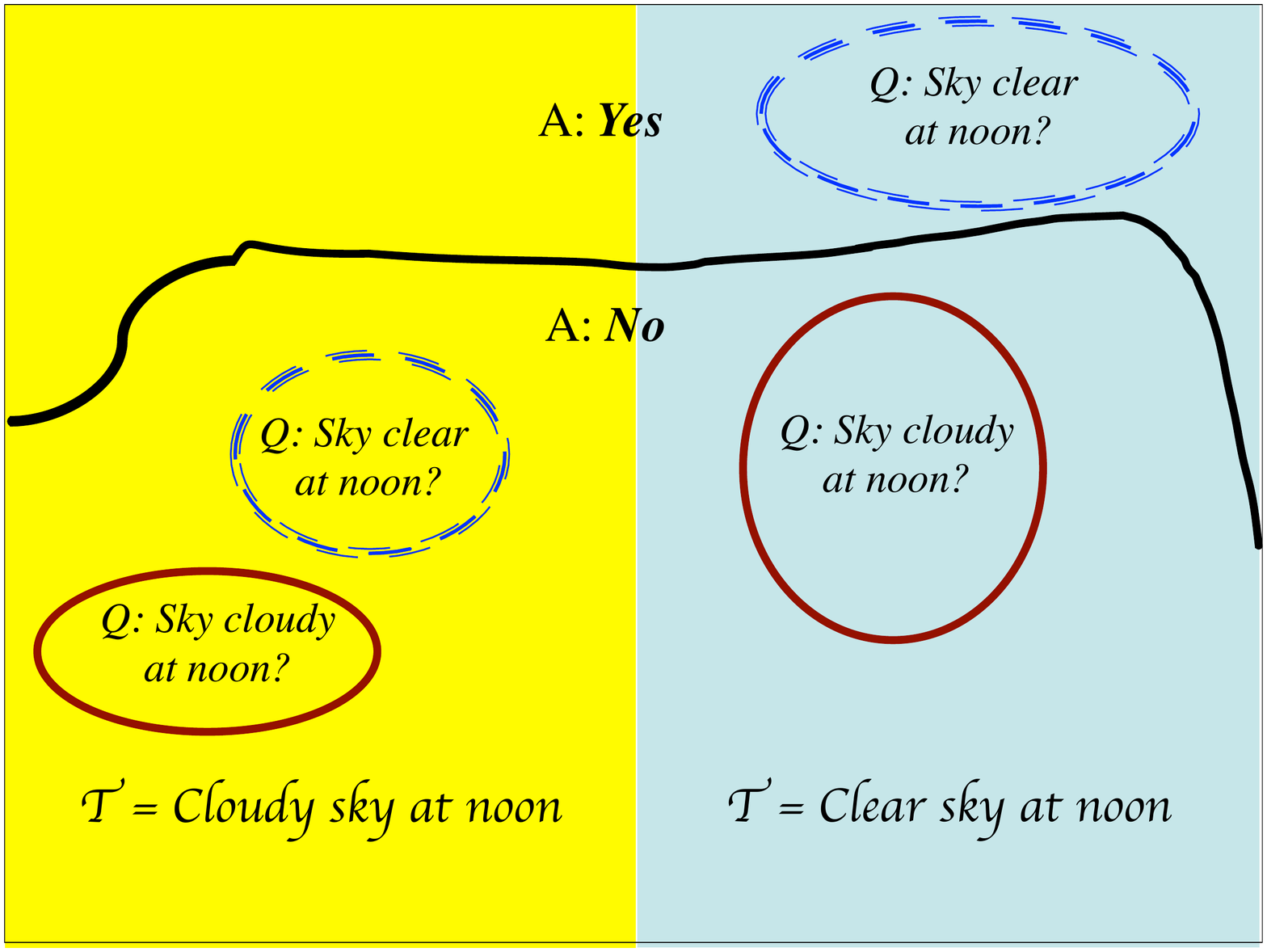}
 	\noindent \caption{An example where the prediction of an ID of the state of the sky at noon cannot be
	guaranteed of being correct, i.e., the ID does not weakly infer the function of $u$ giving the state of
	the sky at noon. The scenario is identical to the one depicted in Fig.~\ref{fig:correct_example},
	except that if the ID is asked the question, ``will the sky be cloudy at noon?'',
and the sky in fact will be cloudy at noon, the ID will answer 'no', which is incorrect. 
	  }
 	\label{fig:incorrect_example}
 \end{figure}
 \label{ex:3}
 \end{example}

\begin{example}
While it is clearly grounded in a real-world scenario, Ex.~\ref{ex:3} obscures the mathematical
essence of weak inference.
A fully abstract, stripped-down example of weak inference is given in the following table,
which provides functions $X(u), Y(u)$ and $\Gamma(u)$ for all $u$ in a space $U$. In this
minimal example, $U$ has only three elements:

\begin{center}
\begin{tabular} {p{1cm} || p{1cm} | p{1cm} | p{1cm} }
  $u$ & $X(u)$ & $Y(u)$ & $\Gamma(u)$   \\ 
\hline \hline a & 1 & 1 & 1        \\
\hline b &  2 & -1     &   1         \\
\hline c &   1 & -1 &  2          \\
\end{tabular}
\label{my-label}
\end{center}

\noindent In this example, $\Gamma(U) = \{1, 2\}$,
so we are concerned with two probes, $\delta_1$ and $\delta_2$.
Setting $X(u) = 2$ means that $u = b$, which in turn means that $\Gamma(u) = 1$
and $Y(u) = -1$. So setting $X(u) = 2$ guarantees that $\Gamma(u) \ne 2$,
and so $\delta_2(\Gamma(u)) = Y(u)$
(which in this case equals -1, the answer 'no'). So the setup value $x = 2$
ensures that the ID correctly answers the binary question, ``does $\Gamma(u) = 2$?'',
in the negative.
Similarly, setting $X(u) = 1$ guarantees that $\delta_1(\Gamma(u)) = Y(u)$,
so that it ensures that the ID correctly answers the binary question, ``does $\Gamma(u) = 1$?'',
in the positive.
\label{ex:4}
\end{example}

Ex.~\ref{ex:4} shows that weak inference can
hold even if $X(u) = x$ doesn't always fix a unique value for $Y(u)$. Such
non-uniqueness is typical when the device is being used for
observation. Setting up a device to observe a variable outside of that
device restricts the set of possible universes; only those $u$ are
allowed that are consistent with the observation device being set up
that way to make the desired observation. But typically just setting
up an observation device to observe what value a variable has doesn't
uniquely fix the value of that variable.

As discussed in App. B of~\cite{wolp08b}, 
the definition of weak inference is very unrestrictive. For example,
a device $\DDD$ is `given credit' for correctly answering probe $\delta(\Gamma(u))$
if there is \emph{any} $x \in X(U)$ such that $ X(u) = x \Rightarrow Y(u) = \delta({\Gamma(u)})$.
In particular, $\DDD$ is given credit even if the binary question 
we would associate with $x$ (under some particular physical interpretation of
what $X$, like in Ex.~\ref{ex:beg_1} and Ex.~\ref{ex:beg_2}) is 
not whether $\Gamma(u) = \gamma$, but some other question. In essence,
the device receives credit even if it gets the right answer by accident.

Unless specified otherwise, a device written as ``${{\DDD}}_i$'' for any
integer $i$ is implicitly presumed to have domain $U$, with setup
function $X_i$ and conclusion function $Y_i$ (and similarly for no
subscript).  Similarly, unless specified otherwise, expressions like
``min$_{x_i}$'' mean min$_{x_i \in X_i(U)}$. 

\subsection{The two Laplace's Demon theorems}
\label{subsec:id_major}

\noindent \emph{``An intellect which at a certain moment would know all forces that set nature in motion, and all positions of all items of which nature is composed, if this intellect were also vast enough  ...  nothing would be uncertain and the future just like the past would be present before its eyes.''}

$ $

\noindent --- Pierre Simon Laplace, ``A Philosophical Essay on Probabilities''

$ $

There are 
limitations on the ability of any device to weakly infer functions. Perhaps
the most trivial is the following:
\begin{proposition}
\label{prop:prop1}
%
%
For any device $\DDD$, there is a function that $\DDD$ does not infer. 
\end{proposition}
\begin{proof}
Choose $\Gamma$ to be the function $Y$, so that the device is trying to infer itself.
Then choose the negation probe $\delta(y \in {\mathbb{B}}) = -y$ to
see that such inference is impossible. (Also see~\cite{wolp08b}.)
\end{proof}

\noindent 

 \begin{figure}
	\hglue-10mm
  	\includegraphics[width= 1.2\linewidth]{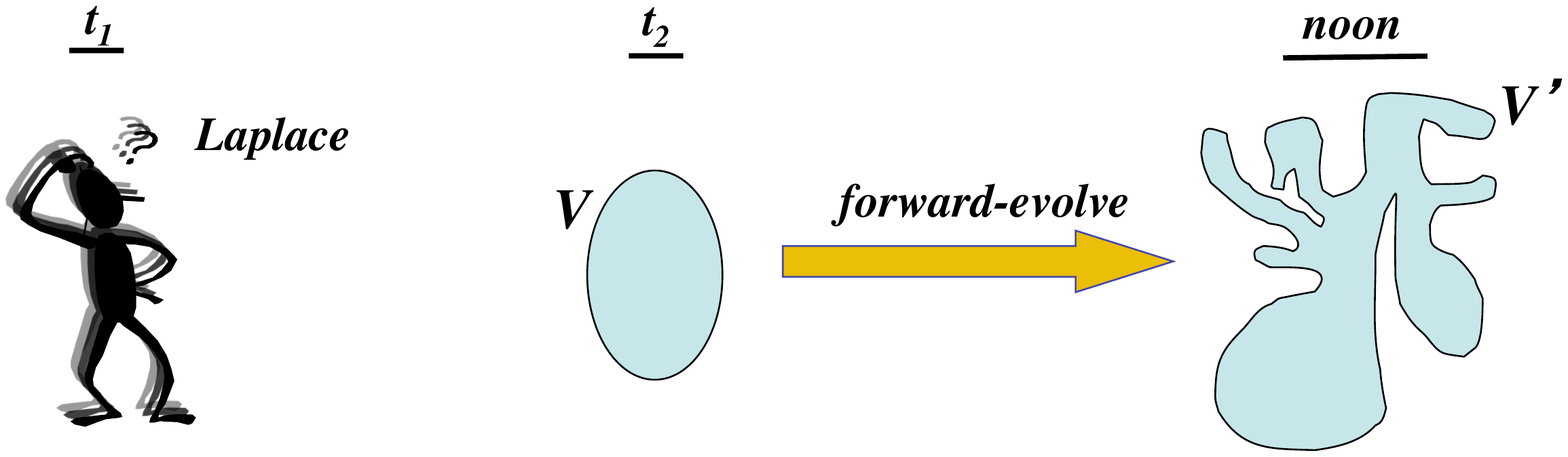}
 	\vglue-35mm
 	\noindent \caption{The time $t_1$ is less than $t_2$, which in turn is less than noon.
    $V$ is the set of all time-$t_2$ universes where Laplace is thinking the answer ``yes'' in response to the $t_1$ question 
    Laplace heard --- whatever that question was.
    $V'$ is $V$ evolved forward to noon. At $t_1$, we ask Laplace, ``will the universe be outside $V'$ at noon?''
    It is impossible for Laplace to answer correctly, no matter what his computational capabilities are, what the laws
    of the universe are, etc. }
 	\label{fig:first_imp}
 \end{figure}

It is interesting to consider the implications of Prop.~\ref{prop:prop1} for the
case where the inference is prediction, as in Ex.~\ref{ex:beg_2}.
Depending on how precisely one interprets Laplace,
Prop.~\ref{prop:prop1} means that he was wrong in his claim
about the ability of an ``intellect''
to make accurate predictions: even if the universe were a
giant clock, it could not contain an intellect that could reliably predict the
universe's future state before it occurred.{\footnote{Similar
conclusions have been reached previously~\cite{mack60,popp88}. However
in addition to being limited to the inference process of prediction,
that earlier work is quite informal. It is no surprise than that some claims in that
earlier work are refuted by well-established results in engineering. For example, the
claim in~\cite{mack60} that ``a prediction concerning the narrator's
future ... cannot ... account for the effect of the narrator's
learning that prediction'' is just not true; it is refuted by adaptive control theory in general and by
Bellman's equations in particular. Similarly, it
is straightforward to see that statements (A3), (A4), and the notion of
``structurally identical predictors'' in~\cite{popp88} have no formal
meaning.}} 
More precisely, for all $\Gamma$ as in Prop.~\ref{prop:prop1}, there could be an intellect
$\DDD$ that can infer $\Gamma$. However
Prop.~\ref{prop:prop1} tells us that for any fixed intellect, there must exist a $\Gamma$
that the intellect cannot infer.  (See Fig.~\ref{fig:first_imp}.) The ``intellect'' Laplace refers
to is commonly called Laplace's ``demon'', so I sometimes
refer to Prop.~\ref{prop:prop1} as the ``first (Laplace's) demon theorem''.

One might think that Laplace could
circumvent the first demon theorem by simply constructing a second demon, specifically designed
to infer the $\Gamma$ that thwarts his first demon. Continuing in this way, one might think that Laplace could construct a set of
demons that, among them, could infer any function $\Gamma$. Then he could construct an ``overseer demon''
that would choose among those demons, based on the function $\Gamma$ that needs to be inferred.
However this is not possible. To see this, simply redefine the device $\DDD$ in
Prop.~\ref{prop:prop1} to be the combination of Laplace with all of his demons.

These limitations on prediction hold even
if the number of possible states of the universe is countable (or even finite),
or if the inference device has super-Turning capabilities.
It holds even if the current formulation of
physics is wrong; it does not rely on chaotic dynamics,
physical limitations like the speed of light, or quantum mechanical
limitations. 

%

Note as well that in Ex.~\ref{ex:beg_2}'s model of a
prediction system the actual values of the times of the various events are not
specified.  So in particular the
impossibility result of Prop.~\ref{prop:prop1} still applies to that example
even if $t_3 < t_2$ --- in which case the time when the agent provides the
prediction is \emph{after} the event they are predicting. Moreover, consider the variant of Ex.~\ref{ex:beg_2} where the
agent programs a computer to do the prediction, as discussed in Footnote~\ref{foot:4}
in that example. In this variant, the program that is input
to the prediction computer could even contain the future value that the agent wants
to predict. Prop.~\ref{prop:prop1} would still mean that the
conclusion that the agent using the computer comes to after
reading the computer's output cannot be guaranteed to be correct.

Prop.~\ref{prop:prop1} tells us that any inference device $\DDD$ can be ``thwarted''
by an associated function. However it does not forbid the possibility
of some second device that can infer that function that thwarts $\DDD$.
To analyze issues of this sort, and more generally to analyze the
inference relationships within sets of multiple functions and multiple
devices, we start with the following definition:
\begin{definition}
\label{def:setup_dist}
Two devices $(X_1, Y_1)$ and $(X_2, Y_2)$ are {\bf{(setup) distinguishable}} iff $\forall x_1,
x_2, \exists u \in U$ such that $X_1(u) = x_1, X_2(u) = x_2$.
\end{definition}
\noindent No device is distinguishable from itself.  Distinguishability is symmetric, but
non-transitive in general (and obviously not reflexive).  

Having two devices be distinguishable
means that no matter how the first device is set up, it is always
possible to set up the second one in an arbitrary fashion; the setting
up of the first device does not preclude any options for setting up
the second one.  Intuitively, if two devices are not distinguishable,
then the setup function of one of the devices is partially
``controlled" by the setup function of the other one.  In such a
situation, they are not two fully separate, independent devices.


I will say that one ID $(X, Y)$ can weakly infer a second one,  $(X', Y')$, if it can weakly infer 
the conclusion of the second ID, $Y'$. (See~\cite{wolp08b} for an example.)
\begin{proposition}
\label{prop:dist_not_infer}
No two distinguishable devices $(X, Y)$ and $(X', Y')$ can weakly infer each other.{\footnote{In fact we can strengthen this result: If $(X', Y')$ can
weakly infer the distinguishable device $(X, Y)$, then $(X, Y)$ can infer neither of the two binary-valued functions equivalent
to $Y'$.}} 
\label{prop:2}
\end{proposition}

\noindent I will call Prop.~\ref{prop:2} the ``second (Laplace's) demon theorem''. 
See Fig.~\ref{fig:second_imp} for an illustration of Prop.~\ref{prop:2}, for two IDs called ``Bob'' and ``Alice'', 
in which they do not directly infer one another's conclusion, but rather infer functions of those
conclusions.

 \begin{figure}
	\hglue-10mm
  	\includegraphics[width= 1.2\linewidth]{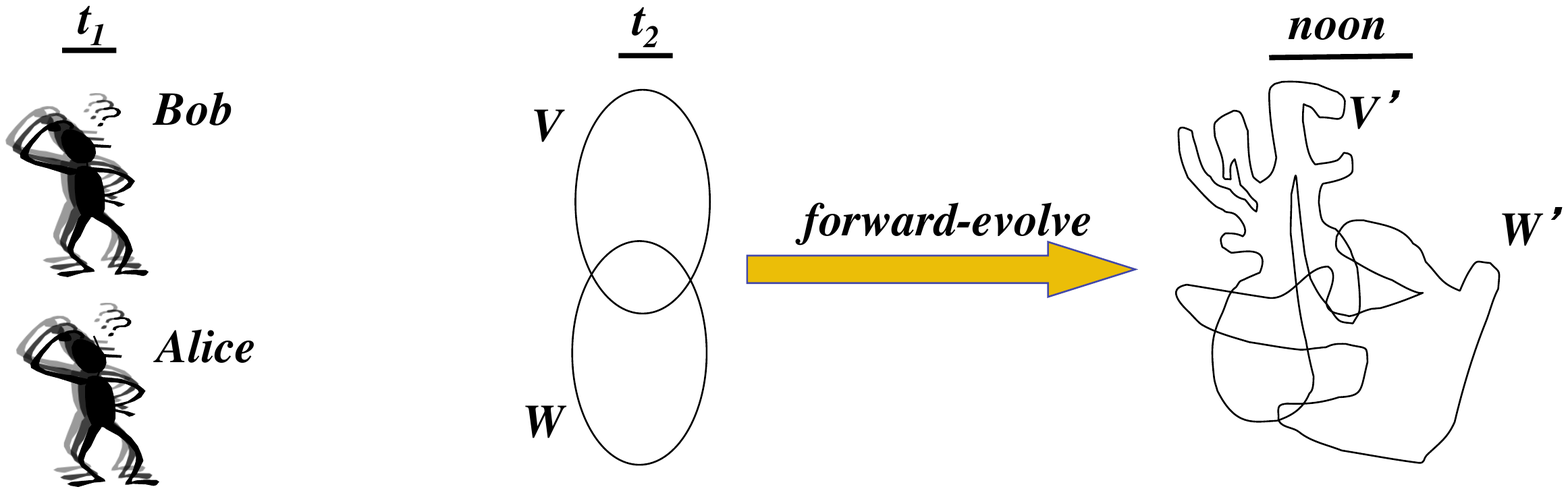}
 	\vglue-35mm
 	\noindent \caption{The time $t_1$ is less than $t_2$, which in turn is less than noon.
    $V$ is the set of all time-$t_2$ universes where Bob is thinking the answer ``yes'' in response to the $t_1$ question 
    Bob heard --- whatever that question was.
      $W$ is the set of all time-$t_2$ universes where Alice is thinking the answer ``yes'' in response to the $t_1$ question 
    Alice heard --- whatever that question was.
    $V'$ is $V$ evolved forward to noon, and $W'$ is $W$ evolved
    forward to noon. At $t_1$, we ask Bob, ``will the universe be in $W'$ at noon?'' (i.e., ``is
    Alice thinking `yes' at $t_2$?''). At that time we also ask Alice, ``will the 
    universe be outside of $V'$ at noon?'' (i.e., ``is Bob \emph{not} thinking `yes' at $t_2$?'').
    It is impossible for both Bob and Alice to answer correctly, no matter what their computational capabilities are, what the laws
    of the universe are, etc. }
 	\label{fig:second_imp}
 \end{figure}

This second Laplace's demon theorem establishes that a
whole class of functions cannot be inferred by $\DDD$ (namely the
conclusion functions of devices that are distinguishable from $\DDD$ and
also can infer $\DDD$). More generally, let $\mathcal{S}$ be a set of devices, all of which are distinguishable
from one another. Then the second demon theorem says that there can be
at most one device in $\mathcal{S}$ that can infer all other devices in
$\mathcal{S}$. It is important to note that the
distinguishability condition is crucial to the second demon theorem; mutual weak
inference can occur between non-distinguishable devices.

In~\cite{barrow2011godel} Barrow speculated 
whether ``only computable patterns are instantiated in physical reality".
There ``computable" is defined in the sense of Turing machine theory.
However we can also consider the term as meaning ``can be evaluated
by a real world computer". If so, then his question is answered --- in the negative ---
by the Laplace demon theorems.

By combining the two demon theorems it is possible to establish the following:

\begin{corollary}
Consider a pair of devices $\DDD = (X, Y)$ and ${{\DDD}'} = (X', Y')$ that are  distinguishable
from one another and whose conclusion functions are inequivalent. Say that ${{\DDD}'}$ weakly 
infers $\DDD$. Then there are at least three 
inequivalent surjective binary functions $\Gamma$ that $\DDD$ does not infer.
\label{coroll:1}
\end{corollary}
\noindent In particular, Coroll.~\ref{coroll:1} means that if any device in a set of distinguishable devices
with inequivalent conclusion functions
is sufficiently powerful to infer all the others, then each of those others must fail
to infer at least three inequivalent functions.

\subsection{Strong inference --- inference of entire functions}
\label{sec:UTM}

As considered in computer science theory,
a computer is an entire map taking an arbitrary ``input'' given
by the value of a physical variable, $\Gamma_1(u)$, to an
``output'' also given by the value of a physical variable,
$\Gamma_2(u)$~\cite{hopcroft2000jd}. It is concerned with saying
how the value of $\Gamma_2(u)$ would change if the value of $\Gamma_1(u)$ changed. 
So it is concerned with two separate physical variables.
In contrast, weak inference is only concerned with inferring the value of a single physical variable, $\Gamma(u)$,
not the relationship between two variables. 

So we cannot really say that a device ``infers a computer'' 
if we only use the weak inference concept analyzed above. 
In this subsection we extend the theory
of inference devices to  include inference of entire functions. In addition to allowing us
to analyze inference of computers, this lays the groundwork for the analysis
in the next section of the relation between inference and algorithmic information theory.

To begin, suppose we have a function $f$ that relates two physical variables.
Since those two variables are themselves functions defined over $U$, in general $f$ is not.
To be more precise, suppose
that there are two function $S$ and $T$ defined over $U$, where $S$ refines $T$, 
and that for all $s \in S(u)$, $f(s) = T(S^{-1}(s))$ is single-valued. We want to define
what it means for a device 
to be able to ``emulate'' the entire mapping taking any $s \in S(U)$ to the associated value 
$f(s) = T(S^{-1}(s))$.

One way to do this is to strengthen the concept of weak inference, so that for any desired
input value $s \in S(U)$, the ID in question can
simultaneously infer the output value $f(s)$ \emph{while also forcing the input to have the value $s$}.
In other words, for any pair $(s \in S(U), t \in T(U))$, by appropriate choice of $x \in X(U)$ the ID $(X, Y)$ simultaneously answers
the probe $\delta_t$ correctly (as in the concept of weak inference) \emph{and}  forces $S(u) = s$.
In this way, when the ID ``answers $\delta_t$ correctly'', it is answering whether $f(s) = t$ correctly,
for the precise $s$ that it is setting. By being able to do this for all $s \in S(U)$, the ID can emulate
the function $f$.

Extending this concept from single-valued functions $f$ to include multivalued functions results in the
following definition:

\begin{definition}
\label{def:defi_5a}
Let $S$ and $T$ be functions both defined over $U$.
A device $(X, Y)$ {\bf{strongly infers}} $(S, T)$ iff
$\forall \; \delta \in {\mathcal{P}}(T)$ and all
$s \in S(U)$, $\exists \; x$ such that $X(u) = x \Rightarrow \{S(u) = s, Y(u) = \delta(T(u))\}$.
\end{definition}
\noindent If $(X, Y)$ strongly infers $(S, T)$ we write $(X, Y) \gg (S, T)$. 

By considering the special case where $T(U) = \B$, we can use strong
inference to formalize what it means for one device to emulate another device:
\begin{definition}
\label{def:defi_5}
A device $(X_1, Y_1)$ {\bf{strongly
infers}} a device $(X_2, Y_2)$ iff $\forall \; \delta \in {\mathcal{P}}(Y_2)$ and all
$x_2$, $\exists \; x_1$ such that $X_1 = x_1 \Rightarrow X_2 = x_2,
Y_1 = \delta(Y_2)$.
\end{definition}
\noindent See App. B in~\cite{wolp08b}
for a discussion of how unrestrictive Def.~\ref{def:defi_5} is.

Def.~\ref{def:defi_5} might seem peculiar, since $(X, Y) \gg (S, T)$ means that
in a certain sense the function $X$ controls what the input to the function $s \rightarrow T(S^{-1}(S))$
is. However, by a simple change in perspective of what device is doing the 
strong inference, we can see that Def.~\ref{def:defi_5} applies even to scenarios that (before
the change in perspective)
do not involve such control. This is illustrated in the following example:

\begin{example}
Suppose ${{\DDD}}_2$ is a device that (for example) can
be used to make predictions about the future state of the weather. Let $\Gamma$ be
the set of future weather states that the device can predict, and let $X_2$ be the set of
possible current meteorological 
conditions. So if this device can in fact infer the future state of the weather,
then for any question $\delta_\gamma$ of whether the future weather will have
value $\gamma$, there is some current condition 
$x_2$ such that if ${{\DDD}}_2$ is set up with that $x_2$, it correctly answers whether
the associated future state of the weather will be $\gamma$. On the other hand, if ${{\DDD}}_2 \not > \Gamma$,
then there is some such question of the form, ``will the future weather be $\gamma$?'' such
that for \emph{no} input to the device of the current meteorological conditions will the device
necessarily produce an answer $y_2$ to the question that is correct.

One way for us to be able to conclude that some device ${{\DDD}}' = (X', Y')$ can ``emulate'' this
behavior of ${{\DDD}}_2$ is to set up ${{\DDD}}_2$ with an arbitrary value $x_2$,
and confirm that ${{\DDD}}'$ can infer the associated value of $Y_2$. So
we require that for all $x_2$, and all $\delta \in \mathcal{P}(Y_2)$, $\exists x'$ such that 
if $X_2 = x_2$ and $X' = x'$, then $Y = \delta(Y_2)$. 

Now define a new device ${{\DDD}}_1$, with its setup function defined by $X_1(u) = (X'(u), X_2(u))$ and its conclusion
function equal to $Y'$. Then our condition for confirming that ${{\DDD}}'$ can emulate ${{\DDD}}_2$
gets replaced by the condition that for all $x_2$, and all $\delta \in \mathcal{P}(Y_2)$, $\exists x_1$ such that 
if $X_1 = x_1$, then $X_2 = x_2$ and $Y = \delta(Y_2)$. This is precisely the definition of strong inference.
\end{example}

Say we have a Turing machine (TM) $T_1$ that can emulate another TM, $T_2$ (e.g., $T_1$
could be a universal Turing machine (UTM), able to emulate any other TM). Such ``emulation'' means that $T_1$ can 
perform any particular calculation that $T_2$ can. The analogous relationship holds for IDs,
if we translate ``emulate'' to ``strongly infer'', and translate ``perform a particular calculation'' to ``weakly infer''.
In addition, like UTM-style emulation (but unlike weak inference), strong
inference is transitive. These results are formalized as follows:

\begin{proposition}
\label{thm:thm_2}
Let ${{\DDD}}_1$, ${{\DDD}}_2$  and ${{\DDD}}_3$ be a set of
inference devices over $U$ and $\Gamma$ a function over $U$. Then:

{\bf{i)}} ${{\DDD}}_1 \gg {{\DDD}}_2$ and ${{\DDD}}_2 > \Gamma$ $\Rightarrow$ ${{\DDD}}_1 >
\Gamma$.

{\bf{ii)}} ${{\DDD}}_1 \gg {{\DDD}}_2$ and ${{\DDD}}_2 \gg {{\DDD}}_3$ 
$\Rightarrow$ ${{\DDD}}_1 \gg {{\DDD}}_3$.
\end{proposition}
\noindent In addition, strong inference implies weak inference, i.e., ${{\DDD}}_1 \gg {{\DDD}}_2
\Rightarrow {{\DDD}}_1 > {{\DDD}}_2$.  

Most of the properties of weak inference have analogs for strong inference:

\begin{proposition}
\label{prop:prop2}
Let ${{\DDD}}_1$ be a device over $U$.

{\bf{i)}} There is a device ${{\DDD}}_2$ such that ${{\DDD}}_1 \not \gg {{\DDD}}_2$.

{\bf{ii)}} Say that $\forall \; x_1$, $|X_1^{-1}(x_1)| >
2$.  Then there is a device ${{\DDD}}_2$ such that ${{\DDD}}_2 \gg {{\DDD}}_1$.
\end{proposition}

\noindent Strong inference also obeys a restriction that is analogous to Prop.~\ref{prop:dist_not_infer}, except that
there is no requirement of setup-distinguishability:
\begin{proposition}
\label{thm:thm_3}
No two devices can strongly infer each other.
\end{proposition}
%

%
%

Recall that there are entire functions that are not computable by any TM, in the sense
that no TM can correctly compute the value of that function for every input to that function. On the other hand, trivially,
any single output value of a function \emph{can} be computed by some TM (just choose the TM that
prints that value and then halts). The analogous distinction holds for inference devices:

\begin{proposition}
Let $U$ be any countable space with at least two elements.
\begin{enumerate}
\item For any function $\Gamma$ over $U$ such that $|\Gamma(U)| \ge 3$ there is a device $\DDD$ that weakly infers $\Gamma$;
\item There is a (vector-valued) function 
$(S, T)$ over $U$ that is not strongly inferred by any device.
\end{enumerate}
\label{prop:whats_inferrable}
\end{proposition}
\begin{proof}
The proof is by construction.
Let $X(u)$ be the identity function (so that each $u \in U$ has its own, unique value $x$).
Choose $Y(u)$ to equal $1$ for exactly one $u$, $\bar{u}$. Then whatever the
value $\gamma := \Gamma(\bar{u}) \in \Gamma(U)$ happens to be, for the probe
$\delta_{\gamma}$ we can choose $x = X(\bar{u})$, so that the device correctly answers `yes' to 
the question of whether $\Gamma(u) = \Gamma(\bar{u})$. For any other probe $\delta_{\gamma'}$, note that
since $|\Gamma(U)| \ge 3$, there must be a $u' \in U$ such that $\Gamma(u') \ne \gamma'$. Moreover, by
construction $Y(u') = -1$. So if we choose $x$ to be $X(u')$, then the device correctly answers `no' to 
the question of whether $\Gamma(u') = \gamma'$. Since this is true for any $\gamma' \ne \Gamma(\bar{u})$,
this completes a proof of the first claim.

We also prove the second claim by construction.
Choose both $S$ and $T$ to be the identity function,
i.e., $S(u) = u$ and $T(u) = u$ for all $u$, so that $|S(U)| = |T(U)| = |U|$. 
So by the first requirement for some device $(X, Y)$ to strongly infer $(S, T)$, it must be that
for any $s$, there is a value of $X$, $x(s)$, such that $X(u) = x(s) \Rightarrow S(u) = s$. 
Since $S$ is a bijection, this means that $x(s)$
must be a single-valued function, for each $s$ choosing a unique ($x$ which in turn chooses a unique) 
$u$. Since $T$ is also a bijection, this means
that $Y(X^{-1}(x(s))$ must equal $1$, in order for the device to correctly answer `yes' to the probe
of whether $T(u) = \delta_{T(S^{-1}(s))}$. However since this is true for all $s \in S(U)$, it is true for all $u \in U$.
So $Y(U)$ is a singleton, contradicting the requirement that the conclusion function of any device be binary-valued.
\end{proof}

\section{Inference in stochastic universes}

\subsection{Stochastic inference}

There are several ways to extend the analysis above to incorporate a
probability measure $P$ over $U$, so that inference is not exact, but
only holds under some probability. In this subsection we present some
of the elementary properties of one such measure of stochastic inference. 

Once there is a distribution over $U$, all functions like $X$, $Y$
and $\Gamma$ become random variables. Now
recall that $\delta_\gamma(\Gamma)$ is shorthand
for the function $u \in U \rightarrow \delta_\gamma(\Gamma(u))$ --- and so now it is 
a random variable. Bearing this in mind, the measure of stochastic inference
we will consider here is defined as follows:

%
\begin{definition}
\label{def:def9}
Let $P(u \in U)$ be a  probability measure and $\Gamma$ a function with
domain $U$ and finite range. Then we say that a device $(X, Y)$
(weakly) infers $\Gamma$ {\bf{with (covariance) accuracy}} 
\begin{eqnarray*}
cov(\DDD, \Gamma) &:=&
\frac{\sum_{\delta \in {\mathcal{P}}(\Gamma)}\max_{x} \big[{\mathbb{E}}_P(Y \delta(\Gamma) \mid x)\big]}{|\Gamma(U)|} 
\end{eqnarray*}
\end{definition}
\noindent Writing it out explicitly, for countable $U$, the numerator in Def.~\ref{def:def9} is
\ba
\sum_{\gamma \in \Gamma(U)} \max_{x \in X(U)} \bigg[ \sum_u Y(u) \delta_\gamma(\Gamma(u)) P( u \mid x) \bigg]
\ea
Intuitively, this is a probe-averaged, best-case (over $x \in X(U)$) probability of answering the probe correctly.


Covariance accuracy is a way to quantify the degree to which $\DDD > \Gamma$ when the inference is subject to uncertainty.
Clearly, $cov(\DDD,\Gamma) \le 1.0$, and if $P$ is nowhere 0, 
then $cov(\DDD,\Gamma) = 1.0$ iff $\DDD > \Gamma$.{\footnote{A
subtlety with the definition of an inference devices arises in this stochastic setting:
we can either require that $Y$ be surjective, as in Def. 1, or instead require that
$Y$ be ``{stochastically surjective}'' in the sense that $\forall y \in {\mathbb{B}}, \; \exists
u$ with non-zero probability such that $Y(u) = y$. The distinction
between requiring surjectivity and stochastic surjectivity of $Y$ will not
arise here.}}
Covariance accuracy obeys the following bound:

\begin{proposition}
\label{prop:cov_lb}
Let $P$ be a probability measure over $U$, 
$\DDD = (X, Y)$ a device, and $\Gamma$ a function over $U$ with finite $|\Gamma(U)|$.
Then
\begin{equation*}
 cov(\DDD,\Gamma) \geq \frac{(2 - |\Gamma(U)|) \max_x  \big[ {\mathbb{E}}_P(Y \mid x ) \big]}{ |\Gamma(U)|}
\end{equation*}
\end{proposition}
 
\begin{proof}  
For any probe $\delta_\gamma$ of  $\gamma \in \Gamma(U)$, 
let $M_\gamma = \max_x \big[ {\mathbb{E}}_P(Y \delta_\gamma(\Gamma) \mid x ) \big]$. 
Define $x_m := \rm{argmax}_x \big[ {\mathbb{E}}_P(Y \mid x ) \big]$. 
Then $M_\gamma \geq  {\mathbb{E}}_P(Y \delta_\gamma(\Gamma) \mid x_m )$ and 
\begin{equation*}
\begin{split}
cov(\DDD,\Gamma)& = \frac{\sum_{\gamma \in \Gamma(U)} M_\gamma}{|\Gamma(U)|} \geq 
	\frac{\sum_{\gamma \in \Gamma(U)} {\mathbb{E}}_P(Y \delta_\gamma(\Gamma) \mid x_m )}{|\Gamma(U)|} \\
	    & =  \frac{ \sum_u P(u \mid x_m)\sum_\gamma Y(u) \delta_\gamma(\Gamma(u))}{|\Gamma(U)|}  \\
	    & = \frac{\sum_u P(u \mid x_m) (2 - |\Gamma(U)|) Y(u)}{|\Gamma(U)|} \\
            &   = \frac{(2 - |\Gamma(U)|) {\mathbb{E}}_P(Y \mid x_m )}{ |\Gamma(U)|}\\
            &   = \frac{ (2 - |\Gamma(U)|) \max_x \big[ {\mathbb{E}}_P(Y\mid x )\big]}{ |\Gamma(U)|}.
\end{split}
\end{equation*}
\end{proof}

\noindent This bound is sharp, as can be seen from the following example.

\begin{example}
Fix some device $\DDD$ and a value $|\Gamma(U)| < \infty$.
Next divide each cell of the partition $X \times Y$ into $|\Gamma(U)|$ parts 
and assign them equal probability. Also map those cells to 1, $\dotso, |\Gamma(U)|$, 
so that $\Gamma(U) = \{1, \dotso, |\Gamma(U)|\}.$
For any given $x \in X$, let $a_x = P(Y = 1 \mid x), b_x = P(Y = -1 \mid x)$. 
For any $x \in X(U), \gamma \in \Gamma(U)$ and associated probe $\delta_\gamma $, 
\begin{equation*}
\begin{split}
\mathbb{E}_P(Y \delta_\gamma(\Gamma) \mid x)& = \frac{a_x + (|\Gamma(U)| - 1) b_x - (|\Gamma(U)| - 1) a_x + b_x}{|\Gamma(U)|}\\
            & = \frac{(2 - |\Gamma(U)|) (a_x - b_x)}{|\Gamma(U)|} = \frac{(2 - |\Gamma(U)|){\mathbb{E}}_P(Y \mid x)}{|\Gamma(U)|}. 
\end{split}
\end{equation*}
We can use this to evaluate
\begin{equation*}
\begin{split}
M_\gamma& := \max_x \big[{\mathbb{E}}_P(Y \delta_\gamma(\Gamma) \mid x)  \big] \\
   & = \frac{(2 - |\Gamma(U)|) \; \max_x \big[ {\mathbb{E}}_P(Y \mid x)\big]}{|\Gamma(U)|}
 \end{split}
\end{equation*}
Since this is the same for all probe parameter values $\gamma$,
\begin{equation*}
cov(\DDD,\Gamma) = \frac{(2 - |\Gamma(U)|) \; \max_x \big[ {\mathbb{E}}_P(Y \mid x)\big]}{ |\Gamma(U)|}
\end{equation*}
which establishes the claim.
\label{ex:cov_lb}
\end{example}

%

The term $\frac{2 - |\Gamma(U)|}{ |\Gamma(U)|}$ in Prop.~\ref{prop:cov_lb} depends only on the size of the space $\Gamma(U)$.{\footnote{Note 
that this term $[2 - |\Gamma(U)|] \;/\; |\Gamma(U)|$ can be negative for $|\Gamma(U)| > 2$. 
This reflects our use of expected values and the convention that $\B = \{-1, 1\}$.}}
The other term, max$_x \big( {\mathbb{E}}_P(Y \mid x) \big)$, can be viewed as a measure of the 
``inference power'' of the device, by analogy with the power of a statistical test.
It quantifies the device's ability to say `yes'.

In the previous section some \emph{a priori} restrictions on the capabilities of IDs 
were presented. These restrictions involved whether certain properties of IDs can(not) be guaranteed with complete certainty.
When we have a probability distribution over $U$ it is appropriate to replace
consideration of ``guaranteed'' properties with consideration of properties that are likely but 
not necessarily guaranteed, e.g., as quantified with covariance accuracy.
When we do that the restrictions of the previous section get modified, sometimes quite substantially. This is illustrated in the next two
propositions.

First, by Prop.~\ref{thm:thm_2}(i), if for devices $\DDD_1$, $\DDD_2$ and function $\Gamma$, $\DDD_1 \gg \DDD_2$ and $\DDD_2 > \Gamma$, 
then $\DDD_1 > \Gamma$. 
In covariance terms, this says that if $\DDD_1 \gg \DDD_2$ and $cov(\DDD_2, \Gamma) = 1.0$, then $cov(\DDD_1, \Gamma) = 1.0$.
What happens to $cov(\DDD_1, \Gamma)$ if $cov(\DDD_2, \Gamma) < 1.0$? A partial answer is given by
the following result:

\begin{proposition}
\label{prop:cov-sinf}
There are devices $\DDD$, $\DDD'$, probability distribution $P$ defined over $U$,
and function $\Gamma$, such that $\DDD' \gg \DDD$ and  
$cov(\DDD, \Gamma)$ is arbitrarily close to 1.0 while $cov(\DDD', \Gamma)$ = 0.
\end{proposition}
\begin{proof}
The proof is by example. 

Let $U$ have ten states, labeled A$, \ldots, $ J and suppose that the functions
$P, \Gamma, \DDD = (X, Y)$ and $\DDD' = (X', Y')$ are
as in Fig.~\ref{fig:table-ws1}, with $0 \le p \le 1$.

\begin{figure}[tbp]
  	\hglue-1.2cm
        \includegraphics[width=1.5\columnwidth]{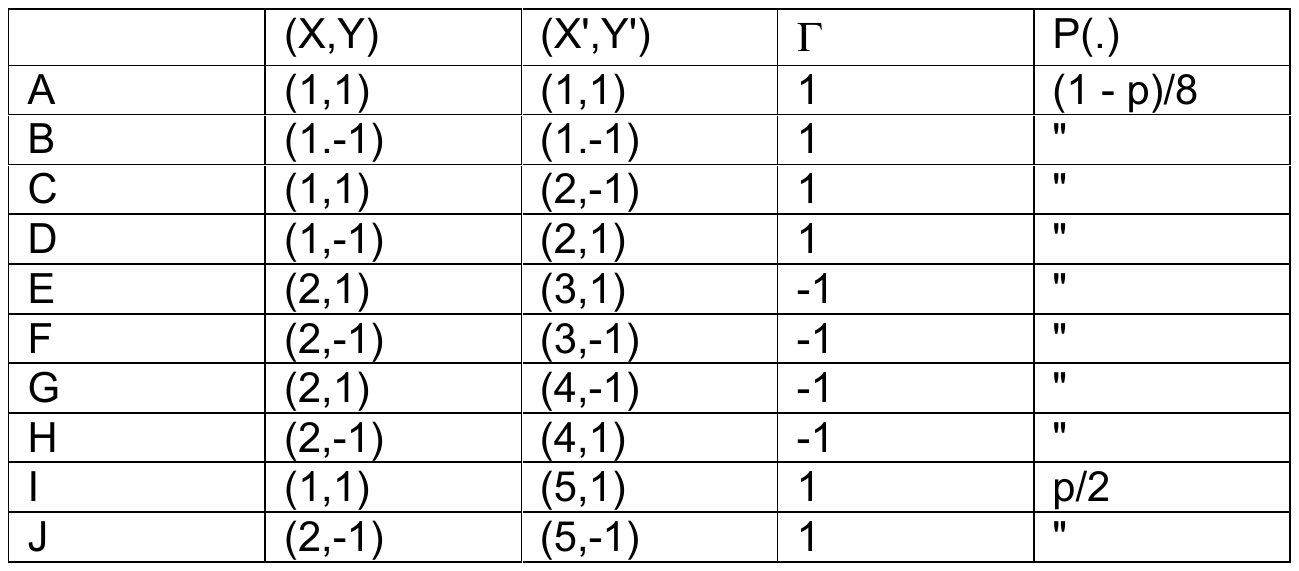}
        \vglue-12cm
        \caption{Specification of a scenario in which the stochastic version of Prop.~\ref{thm:thm_2}(i),
        concerning ``transitivity'' of weak inference through strong inference, fails drastically.}
\label{fig:table-ws1}
\end{figure}
%
\begin{enumerate}
\item To verify that $\DDD' \gg \DDD$,
for the $1$-probe, for $x = 1,2$, choose $x' = 1,3$, respectively. For the $-1$-probe, for $x = 1,2$, choose $x' = 2,4$,
respectively.

\item $cov(\DDD, \Gamma) = p$. To see this,
for the $1$-probe, evaluate $\max_x {\mathbb{E}}_P(Y \delta_1(\Gamma) \mid x) = p$, the maximum
occurring for $x = 1$.
Similarly, for the $-1$-probe, evaluate $\max_x {\mathbb{E}}_P(Y \delta_{-1}(\Gamma) \mid x) = p$, the maximum
occurring for  $x = 2$.

\item $cov(\DDD', \Gamma) = 0$. To see this
for both probes, note that ${\mathbb{E}}_P(Y' \delta(\Gamma) \mid x') = 0$ for each $x'$.
\end{enumerate}
The proof is completed by taking $p \rightarrow 1$.
\end{proof}

%


To understand Prop.~\ref{prop:cov-sinf},
recall that the definition of $\DDD' \gg \DDD$ requires that for any $x \in X(U)$ and for any probe $\delta_\gamma \in {\mathcal{P}}(\Gamma)$, there be \emph{some} $x'$ and associated $X'^{-1}(x') \subseteq U$ for which $\DDD'$ successfully emulates $\DDD$'s
behavior at inferring $\delta_\gamma$. If the inference  $\DDD > \Gamma$ is perfect, then 
$\DDD'$ also infers $\Gamma$.
However, if the inference  $\DDD > \Gamma$ is only partially correct, then that 
value $x'$ and associated subset of $U$, under which $\DDD' \gg \DDD$ may be 
precisely those $u$ for which $\DDD$ performs badly at inferring $\delta_\gamma$.
Thus, $\DDD$ may do an excellent, though imperfect, job overall of inferring $\Gamma$ while $\DDD'$ fails completely.

The second example of how the restrictions of the previous section get modified
by introducing a probability distribution is that this makes the second Laplace's impossibility theorem 
become ``barely true'':
\begin{proposition}
\label{prop:non-dist-cov}
There are devices $\DDD$ and $\DDD'$ with $X$ and $X'$ setup-distinguishable and
a distribution $P$ where both $cov(\DDD, \DDD')$ and $cov(\DDD', \DDD)$ are arbitrarily close to 1.
\end{proposition}

\begin{proof}
The proof is by example. 

Let $U$ have sixteen states, labeled A, $\ldots$, P and suppose that the functions
$P, \Gamma, \DDD = (X, Y)$ and $\DDD' = (X', Y')$ are
as in Fig.~\ref{fig:table-ww1}, with arbitrary $0 < b < 1/6$, and $a = (1 - 6b)/2$.

By inspection, $X$ and $X'$ are setup distinguishable. Next, plugging in yields
$cov(\DDD,\DDD') = cov(\DDD, Y') = a/(a + b)$. Moreover $cov(\DDD',\DDD) = cov(\DDD,\DDD') $ by symmetry
of the columns in Fig.~\ref{fig:table-ww1}.
(${\mathbb{E}}_P(Y  Y' \mid X = 1) = (a - b)/(a + b)$ and ${\mathbb{E}}_P(Y Y' \mid X = -1) = -1$.)

\begin{figure}[tbp]
  	\hglue-1.2cm
        \includegraphics[width=1.7\columnwidth]{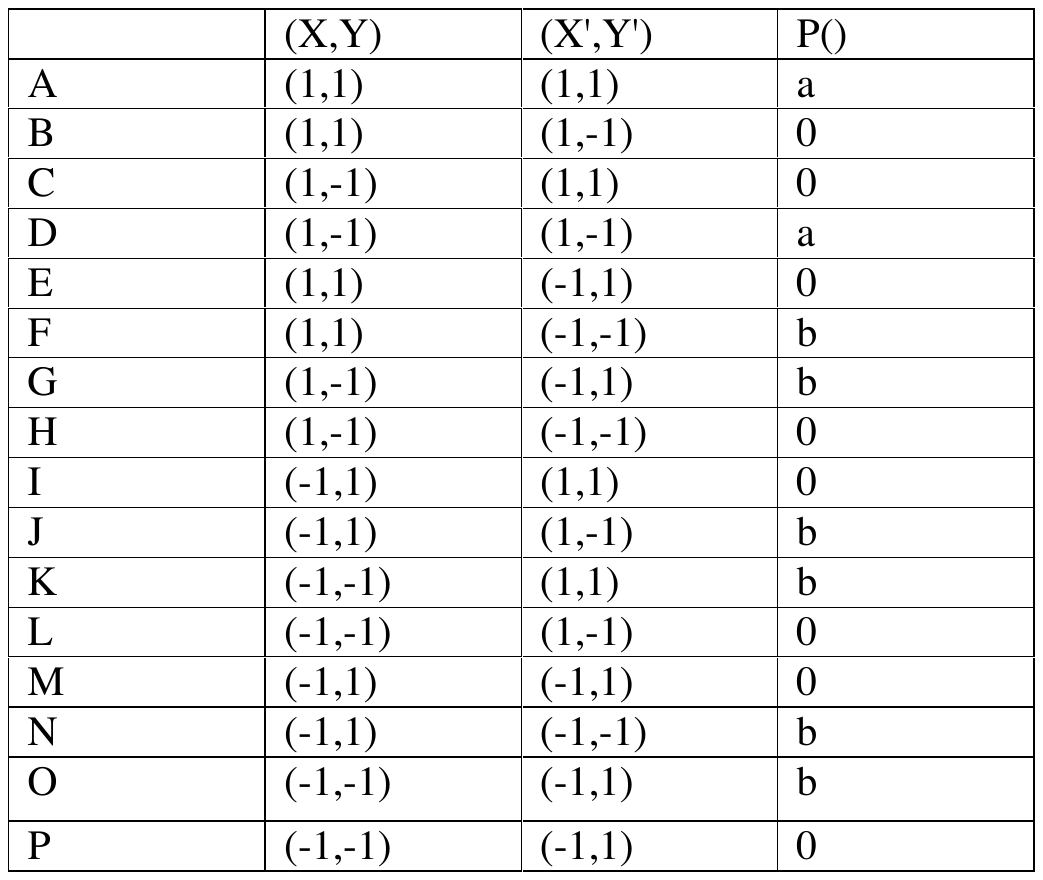}
        \vglue-11cm
        \caption{Specification of a scenario in which the stochastic version of Prop.~\ref{prop:dist_not_infer},
        concerning simultaneous inference of two setup-distinguishable IDs, fails drastically.}
\label{fig:table-ww1}
\end{figure}
%

So by taking $b$ arbitrarily close to 0, both of the covariances can be made arbitrarily close to 1.
\end{proof}

Prop.~\ref{prop:non-dist-cov} shows that in a certain sense, as soon as any stochasticity is introduced into the universe,
having two devices be setup-distinguishable no longer restricts their ability
to simultaneously infer each other. However
if we replace setup-distinguishability with the property that the setup functions of the two devices are statistically
independent, then we recover strong restrictions on simultaneous inference.
%

To illustrate this, let $M$ be the four-dimensional hypercube $\{0, 1\}^4$. Define the following
three functions over $\vec{z} \in M$:
\begin{enumerate}
\item $k({\vec{z}}) = z_1 + z_4 - z_2 - z_3$;
\item $m({\vec{z}}) = (z_2 - z_4)$;
\item $n({\vec{z}})=(z_3 - z_4)$. 
\end{enumerate}

\begin{proposition}
\label{prop:prop6}
Let $P$ be a probability measure over $U$,
and ${{\DDD}}_1$ and ${{\DDD}}_2$ two devices 
where $X_1(U) = X_2(U) = {\mathbb{B}}$, and those variables are
statistically independent under $P$. Define
$P(X_1 = -1) \equiv \alpha$ and $P(X_2 = -1)
\equiv
\beta$. Say that ${{\DDD}}_1$ infers ${{\DDD}}_2$ with accuracy $\epsilon_1$, while
${{\DDD}}_2$ infers ${{\DDD}}_2$ with accuracy $\epsilon_2$. Then
\begin{eqnarray*}
\epsilon_1 \epsilon_2 \;&\le&\; {\mbox{max}}_{{\vec{z}} \in M} 
\big| \alpha \beta [k({\vec{z}})]^2 + \alpha k({\vec{z}})m({\vec{z}}) +
\beta k({\vec{z}})n({\vec{z}}) + m({\vec{z}})n({\vec{z}})\big| .
\end{eqnarray*}
In particular, if $\alpha = \beta = 1/2$, then
\begin{eqnarray*}
\epsilon_1 \epsilon_2 \;&\le&\; \frac{{\mbox{max}}_{{\vec{z}} \in M}
\; |\; (z_1 -
z_4)^2 - (z_2 - z_3)^2\; |}{4} \nonumber \\
&=&\; 1/4.
\end{eqnarray*}
\end{proposition}


The maximum for $\alpha = \beta = 1/2$ can occur in several
ways. One is when $z_1 = 1$, and $z_2, z_3, z_4$ all equal $0$. At
these values, both devices have an inference accuracy of 1/2 at
inferring each other. Each device achieves that accuracy by perfectly
inferring one probe of the other device, while performing randomly for
the remaining probe.

The ID framework as developed to date has no
function measuring distance, nor one measuring time. So at present, one
cannot even formulate an ID-analog of Heisenberg's uncertainty principle, never
mind try to derive it. It is intriguing that despite this,  Prop.~\ref{prop:prop6} 
is a bound on the product of uncertainties, exactly like Heisenberg's uncertainty principle.
This suggests it may be worth exploring extensions of the ID framework
that do involve distance and time, to see what \emph{a priori} constraints
there might be on the product of uncertainties of two IDs that are measuring
different aspects of the same system. (This idea is returned to in the last section below.)

%

Finally, it should be noted that there are other ways to quantify the degree of weak inference
when there is intrinsic uncertainty, in addition to covariance accuracy. For 
example, we could
change Def.~\ref{def:def9} by replacing the sum over all probes $\delta$ and associated
division by $|\Gamma(U)|$ with a minimum over all probes $\delta$. (This amounts to 
replacing an average-best-case expression with a worst-case expression.)

\subsection{The complexity of inference} 

Constraints on
what can be computed by a physical device can be derived from the laws of physics~\cite{lloyd2000ultimate}.
There have also been attempts to
go the other way, and derive constraints on the laws of physics from computation theory, in 
particular from algorithmic information theory (AIT)~\cite{livi08,chaitin2004algorithmic,
zure89a,zure89b,zurek1990complexity,zenil2012computable}. These often implicitly 
involve uncertainty about the state of the universe. For example, the use of Kolmogorov
complexity to model physical reality is often intimately related to the use of algorithmic probability~\cite{livi08,schmidhuber2000algorithmic,zuse1969rechnender}.
(Indeed, the very first line in~\cite{schmidhuber2000algorithmic}
is ``The probability distribution $P$ from which the history of our universe is sampled
represents a theory of everything''.) One way to justify consideration of such a
probability distribution in the first place is to identify
it with uncertainty of some agent  (e.g., a scientist) concerning the state of the universe. 

This importance of an agent in attempts to analyze physics using AIT
suggests we extend the inference device framework to include structures similar
to those considered in AIT.
There are several ways to extend the ID framework this way. In this subsection I sketch the starting point
for one of them.
%

%
%

\label{sec:inf__compl}

Given a TM $T$, the \emph{Kolmogorov complexity} of an
output string $s$ is defined as the size of the smallest input
string $s'$ that when input to $T$ produces $s$ as output. To
construct our inference device analog of this, we need to define the
``size'' of an input region of an inference device $\DDD$. To do this,
we assume we are given a measure $d\mu$ over $U$, and for simplicity
restrict attention to functions $\Gamma$ over $U$ with countable
range. Then we define the {\bf{size}} of $\gamma \in \Gamma(U)$ as 
-ln$\big[\int_{\Gamma^{-1}(\gamma)}
d\mu(u) \; 1\big]$, i.e., the negative logarithm of the measure of all
$u \in U$ such that $\Gamma(u) = \gamma$.{\footnote{As usual, if $U$ is countable, 
$\mu$ is a point measure, and the integral is a sum.}} We write this size as
${\mathcal{M}}_{\mu; \Gamma}(\gamma)$, or just ${\mathcal{M}}(\gamma)$ for
short.{\footnote{If $\int d\mu(u) \; 1 = \infty$, then we instead work
with differences in logarithms of volumes, evaluated under an
appropriate limit of $d\mu$ that takes $\int d\mu(u) \; 1 \rightarrow
\infty$. For example, we might work with such differences when $U$ is
taken to be a box whose size goes to infinity. 
\label{foot:vol}}}  

We define inference complexity in terms of such a size
function using the shorthand introduced just below Eq.~\eqref{eq:shorthand1}:
\begin{definition}
\label{def:def6}
Let $\DDD$ be a device and $\Gamma$ a function over
$U$ where $X(U)$ and $\Gamma(U)$ are countable and $\DDD > \Gamma$. The
{\bf{inference complexity}} of $\Gamma$ with respect to $\DDD$ and measure $\mu$ is defined
as
\begin{eqnarray*}
{\mathcal{C}}_\mu(\Gamma ; \DDD) \;\;&\triangleq& \;\; \sum_{\delta \in {\mathcal{P}}(\Gamma)}
	{\mbox{min}}_{x : X = x \Rightarrow Y = \delta(\Gamma)}
			 [{\mathcal{M}}_{\mu,X} (x)].
\end{eqnarray*}
\end{definition}

\noindent In the sequel I will often have the measure implicit, and 
(for example) simply write ${\mathcal{C}}$
rather than ${\mathcal{C}}_\mu$. I will also mostly restrict attention to the case
where $\mu$ is either a distribution or a semi-measure.\footnote{A natural alternative 
measure of ``inference complexity'' is given by replacing
the sum over all probes in Def.~\ref{def:def6} with a max over all probes, so that we are
analyzing the hardest possible question to ask about $\Gamma$. In the interests of
space, we leave this for future work.}

As an example, for the case where inference models the process of prediction, $\Gamma$ corresponds to a potential future
state of some system $S$ external to $\DDD$. In this case ${\mathcal{C}}(\Gamma;
\DDD)$ is a measure of how difficult it currently is for $\DDD$ to predict that future
state of $S$. Loosely speaking, the more sensitively that future state
depends on current conditions, the greater the inference complexity of predicting
that future state.

Inference complexity of any function $\Gamma$ with respect to a device $(X, Y)$  is bounded by the Shannon entropy of $\mu(X)$:

\begin{proposition}
For any ID $\DDD$, probability distribution $\mu$, and function $\Gamma$ with a countable image such that $\DDD > \Gamma$,
\ba
{\mathcal{C}}_\mu(\Gamma ; \DDD) \le  |\Gamma| \times H_\mu(X) \nonumber
 \ea
 where $H_\mu(X)$ is the Shannon entropy of $\mu(X)$.
\label{prop:inf_comp_entropy}
\end{proposition}

\begin{proof} 
Expand
\begin{eqnarray*}
 \sum_{\delta \in \mathcal{P}(\Gamma)}
	\min_{x : X = x \Rightarrow Y = f(\Gamma)}
			 [{\mathcal{M}}_{\mu,X} (x)] &\le& \sum_{x\in X(U)} {\mathcal{M}}_{\mu,X} (x) \\
			 &\le& -|\Gamma| \sum_{x\in X(U)}  \frac{ {\mbox{log}}_2 \mu(x)}{|\Gamma|} \\
			 &\le& |\Gamma| H_\mu(X)
\end{eqnarray*}
\end{proof}

Kolmogorov complexity concerns TMs computing a single output, rather than TMs emulating
an entire function from inputs to outputs. The field of algorithmic information theory then
analyzes the relation between Kolmogorov complexity and UTMs, i.e., TMs that emulate entire functions from 
inputs to outputs. Analogously, inference complexity concerns inferring a single value of
a variable, i.e., it is defined in terms of \emph{weak} inference. So to
investigate the inference device analog of algorithmic 
information theory means investigating the relation between inference complexity and 
IDs that emulate entire functions --- which involves strong inference instead of weak inference.

To begin, recall  perhaps the most fundamental result in AIT, the
\emph{invariance theorem}. This theorem gives an upper bound on the
difference between the Kolmogorov complexity of a string using a
particular UTM $T_1$ and its complexity if using a different UTM,
$T_2$. This bound is independent of the computation to be performed,
and can be viewed as the Kolmogorov complexity of $T_1$ emulating
$T_2$. Similarly, we can bound how much greater the inference complexity
of a function can be for a device ${\mathcal{D}}_1$ than it is for a different
device ${\mathcal{D}}_2$ if ${\mathcal{D}}_1$ can strongly infer ${\mathcal{D}}_2$:

\begin{proposition}
\label{thm:thm4}
Let ${\mathcal{D}}_1$ and ${\mathcal{D}}_2$ be two  devices and
$\Gamma$ a function over $U$ where $\Gamma(U)$ is finite, ${\mathcal{D}}_1 \gg
{\mathcal{D}}_2$, and ${\mathcal{D}}_2 > \Gamma$. Then for any distribution $\mu$,
\begin{eqnarray*}
&&{\mathcal{C}}_\mu(\Gamma ; {\mathcal{D}}_1) - {\mathcal{C}}_\mu(\Gamma ; {\mathcal{D}}_2) 
			\;\;\le \;\; |\Gamma(U)| \; \times \nonumber \\
&& \qquad \qquad  {\max}_{x_2} \bigg( \min_{x_1 : \{X_1 = x_1
   \Rightarrow X_2 = x_2, Y_1 = Y_2\}} [{\mathcal{M}}_{\mu,X_1}(x_1) -
{\mathcal{M}}_{\mu,X_2}(x_2)]  \bigg).
\end{eqnarray*}
\noindent 
\end{proposition}

\noindent Note that since ${\mathcal{M}}_{\mu,X_1}(x_1) - {\mathcal{M}}_{\mu,X_2}(x_2) =
{\mbox{ln}}\bigg[\frac{\mu (X_2^{-1}(x_2))} {\mu( X_1^{-1}(x_1))} \bigg]$, the bound in
Prop.~\ref{thm:thm4} is independent of the units with which one measures volume in
$U$. (Cf. footnote ~\ref{foot:vol}.)  Furthermore, it is always true that $X_1 =
x_1 \Rightarrow X_2 = x_2, Y_1 = Y_2$ iff $X_1^{-1}(x_1) \subseteq
X_2^{-1}(x_2) \; \cap \; (Y_1Y_2)^{-1}(1)$. 
Accordingly, for all $(x_1, x_2)$ pairs arising in the bound in
Prop.~\ref{thm:thm4}, $\frac{\mu (X_2^{-1}(x_2))} {\mu( X_1^{-1}(x_1))} \ge 1$. So the upper bound in
Prop.~\ref{thm:thm4} is always non-negative.

The max-min expression on the RHS of Prop.~\ref{thm:thm4} is independent of $\Gamma$.
So the bound in Prop.~\ref{thm:thm4} 
is independent of all 
aspects of $\Gamma$ except the cardinality of $\Gamma(U)$. 
Intuitively, the bound is $|\Gamma(U)|$ times the worst-case amount of
``computational work'' that ${\mathcal{D}}_1$ has to do to ``emulate'' ${\mathcal{D}}_2$'s
behavior for some particular value of $x_2$. 

Suppose that it takes a lot of computational work
for $\DDD_2$ to infer $\Gamma$, and so it also takes a lot of computational work for $\DDD_1$ 
to infer $\Gamma$ by emulating $\DDD_2$.
However, it might take very little work for $\DDD_1$ to infer $\Gamma$ directly.
In fact, it may even be that ${\mathcal{C}}(\Gamma ; {\mathcal{D}}_1) < {\mathcal{C}}(\Gamma ; {\mathcal{D}}_2)$:

 \begin{proposition}
 \label{prop:sinf-sic}
 There are devices $\DDD$, $\DDD'$, probability distribution $P$ defined over $U$,
and function $\Gamma$, such that 
 $\DDD > \Gamma$, $\DDD' \gg \DDD$, and ${\mathcal{C}}_P(\Gamma; \DDD)$ 
 is arbitrarily large, while ${\mathcal{C}}_P(\Gamma; \DDD')$ is arbitrarily close to the minimum value of 
 $\big|\Gamma \big| \times \ln(|\Gamma(U)|)$.
 \end{proposition}
 
 \begin{proof} 
 The proof is by example. 

Let $U$ have twelve states, labeled A$, \ldots, $ L and suppose that the functions
%
 $P, \Gamma, \DDD' = (X, Y)$ and $\DDD = (X, Y)$ are
as in Fig.~\ref{fig:table-sic-s}, with $1/4 < p < 1$.

\begin{figure}[tbp]
  	\hglue-3cm
        \includegraphics[width=1.7\columnwidth]{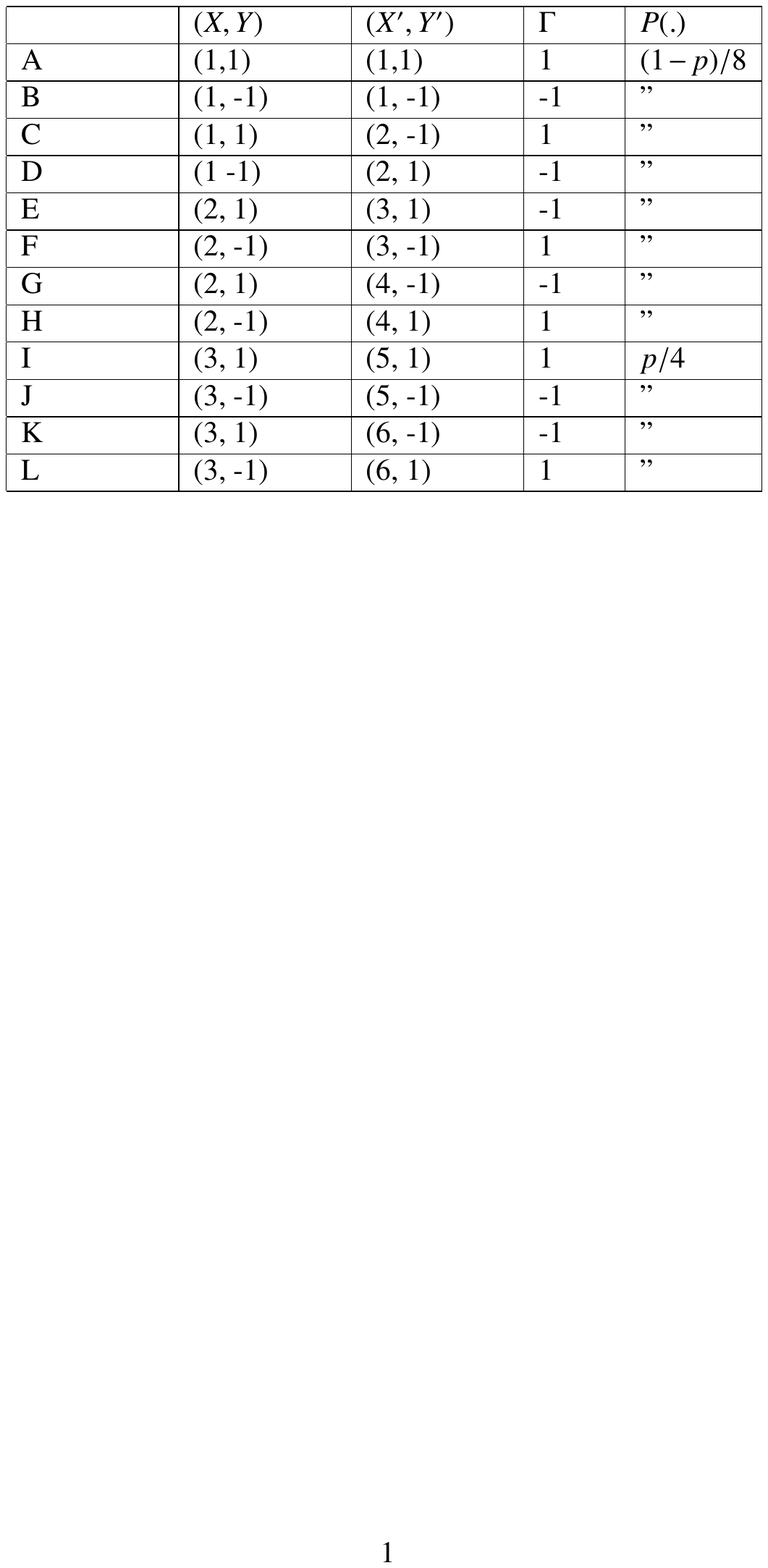}
        \vglue-12.5cm
        \caption{Scenario illustrating discrepancies of complexities of two IDs where one strongly infers the other.}
\label{fig:table-sic-s}
\end{figure}

\begin{enumerate}
\item To verify that $\DDD > \Gamma$, for the $1$-probe, choose $x = 1$. For the $-1$ probe,
choose $x = 2$.

\item To verify that $\DDD' \gg \DDD$, first, for the $1$-probe, for $x = 1,2,3$, choose $x' = 1,3,5$, respecitively.
Then for the $-1$-probe, for $x = 1,2,3$, choose $x' = 2,4,6$, respectively.

\item To verify that ${\mathcal{C}}(\Gamma ; \DDD)$ can be arbitrarily large,
first expand it as $-2\ln((1 - p)/2) = 2\ln(2) - 2\ln(1 - p)$.
(For the 1-probe, $x = 1$ and ${\mathcal{M}}_{P, X}(x) = -\ln((1 - p)/2)$ and similarly for the -1-probe and $x = 2$.)

\item To verify that ${\mathcal{C}}(\Gamma ; \DDD')$ can be arbitrarily close to its minimal value,
write it as $-2\ln(p/2) = 2\ln(2) - 2\ln(p)$.
(For the 1-probe, $x' = 5$ and ${\mathcal{M}}_{P,X}(x) = -\ln(p/2)$ and similarly for the -1-probe and $x' = 6$.)
\end{enumerate}

\noindent Finally, by taking $p$ arbitrarily close to 1, ${\mathcal{C}}(\Gamma ; \DDD)$ becomes arbitrarily large while 
${\mathcal{C}}(\Gamma ; \DDD')$ becomes arbitrarily close to the minimum of $2\ln(2)$.
\end{proof}

Although there is not space to analyze them here, it is worth noting that
there are several ways to translate some of the
mathematical structures of algorithmic information theory into the
inference device framework. For example, just as a given Turing machine may fail
to produce an output for some specific input, so an inference device may fail
to reach a conclusion for some specific setup. This motivates the following definition:

\begin{definition}
A device $(X, Y)$ \textbf{halts} for setup value $x$ iff $X= x \Rightarrow Y= y$ for some
single value $y$.
\end{definition}
\noindent We say that $x$ is a ``halting setup'' if $(X, Y)$  halts
for $x$. Parelleling the usual definitions in TM theory, 
we say that an ID is \textbf{total}, or \textbf{recursive} iff it
halts for all $x \in X(U)$. So an ID $(X, Y)$ is recursive iff $X$ refines $Y$.

Given this definition of what it means for a device to halt on a given input, we can
define the inference analog of a prefix-free Turing machine~\cite{livi08}:
\begin{definition}
Given a semi-measure $\mu$, a device $(X, Y)$ is \textbf{prefix(-free)} iff
\ba
\sum_{x : \DDD \; halts \; on \; x} 2^{-{\mathcal{M}}_{\mu,X} (x)} &\le& 1 \nonumber
\ea
\end{definition}
\noindent By Kraft's inequality, if $\DDD$ is prefix-free for a semi-measure $\mu$, 
then there is a prefix-free code for the set of all halting $x \in X(U)$. Therefore we
can identify that set of $x$'s with semi-infinite bit strings, or equivalently with the
natural numbers~\cite{livi08}.

As a final example, note that the min over $x$'s in Def.~\ref{def:def6} is a direct analog of the min in the
definition of Kolmogorov complexity (there the min is over those
strings that when input to a particular UTM result in the desired
output string). A natural modification to Def.~\ref{def:def6} is to remove the
min by considering all $x$'s that cause $Y = \delta(\Gamma)$, not just of
one of them:
\begin{eqnarray*}
{\hat{{\mathcal{C}}}}(\Gamma ; \DDD) \;\;&\triangleq& \;\; \sum_{\delta \in
{\mathcal{P}}(\Gamma)} -{\mbox{ln}} \left[\; \mu \left(\cup_{x :  X
= x \Rightarrow Y = \delta(\Gamma)} X^{-1}(x) \right) \;\right]  \nonumber \\
&=& \sum_{\delta \in {\mathcal{P}}(\Gamma)} -{\mbox{ln}} \left[\sum_{x : X = x \Rightarrow
Y = \delta(\Gamma)} e^{-{\mathcal{M}}(x)}\right],
\end{eqnarray*}
where the equality follows from the fact that for any $x, x' \ne x$,
$X^{-1}(x) \cap X^{-1}(x') = \varnothing$.  The argument of the $\ln(.)$ in
this modified version of inference complexity has a direct analog in
TM theory: The sum, over all input strings $s$ to a UTM that generates
a desired output string $s'$, of $2^{-n(s)}$, where $n(s)$ is the bit
size of $s$. This is sometimes known as the ``algorithmic'' or ``Solomonoff''
probability of $s'$~\cite{livi08} in the theory of TMs.

\section{Modeling the physical universe in terms of inference devices}

I now expand the scope of the discussion to allow sets of many  inference devices and / or many
functions to be inferred. Some of the
philosophical implications of the ensuing results are then discussed in the next subsection.

\subsection{Formalization of physical reality involving Inference Devices}
\label{sec:realities}

Define a {\bf{reality}} as a pair $(U; \{F_\phi\})$ where the space $U$ is 
the {\bf{domain}} of the reality, and
$\{F_\phi\}$ is a (perhaps uncountable) non-empty set of functions all
having domain $U$. We are particularly interested in {\bf{device
realities}} in which some of the functions are binary-valued, and we
wish to pair each of those functions uniquely with some of the other
functions. In general, not all of the functions in $\{F_\phi\}$
need to be members of such a pair. Accordingly, 
the most general form of such realities is triples of the form $(U;
\{(X_\alpha, Y_\alpha)\}; \{\Gamma_\beta\})$, or just $(U; \{{{\DDD}}_\alpha\};
\{\Gamma_\beta\})$ for short, where $\{{{\DDD}}_\alpha\}$ is a set of devices over $U$
and $\{\Gamma_\beta\}$ a set of functions over $U$.

Define a {\bf{universal device}} as any device in a reality that can
strongly infer all other devices and weakly infer all functions in
that reality. Prop.~\ref{thm:thm_3} means that no reality can contain more than one
universal device. So in particular, if a reality contains a
universal device and there is a given distribution over $U$, then 
the reality has a unique natural choice for an inference
complexity measure, namely the inference complexity with respect to
its (unique) universal device. (This contrasts with Kolmogorov
complexity, which depends on the arbitrary choice of what UTM to use.)

For simplicity, assume the index set $\phi$ is countable, with elements $\phi_1, \phi_2, \ldots$.
It is interesting to consider the {\bf{reduced form}} of a reality $(U;
\{F_\phi\})$, which is defined
as the image of the function $u \rightarrow (F_{\phi_1}(u), F_{\phi_2}(u), \ldots)$.
In particular, the reduced form of a
device reality is the set of all tuples $([x_1, y_1], [x_2, y_2],
\ldots; \gamma_1, \gamma_2, \ldots)$ for which $\exists \; u \in U$
such that simultaneously $X_1(u) = x_1, Y_1(u) = y_1, X_2(u) = x_2,
Y_2(u) = y_2, \ldots ; \Gamma_1(u) =
\gamma_1, \Gamma_2(u) = \gamma_2, \ldots$. By working with reduced forms of
realities, we dispense with the need to explicitly discuss $U$ entirely.{\footnote{Note
the implication that if we work with reduced realities,
all of the non-stochastic analysis of the previous sections can be reduced to
satisfiability statements concerning sets of categorial variables. For example, the fact
that a device cannot weakly infer itself is equivalent to the statement that there
is no countable space $X$ with at least two elements and associated
set of pairs $\mathcal{V} = \{(x_i, y_i)\}$ where all $y_i \in \B$, 
such that for both probes $\delta$ of $y_i$, there
is some value $x' \in X$ such that in all pairs $(x', y) \in \mathcal{V}$, $y = \delta(y)$.}}

\begin{example}
Take $U$ to be the set of all possible histories of a universe across all time that are consistent
with the laws of physics. So each $u$ is a specification of a trajectory
of the state of the entire universe through all time. The laws of physics are then
embodied in restrictions on $U$. For example, if
one wants to consider a universe in which the laws of physics are
time-reversible and deterministic, then we require that no two
distinct members of $U$ can intersect. Similarly, properties like
time-translation invariance can be imposed on $U$, as can more
elaborate laws involving physical constants. 

Next, have \{$\Gamma_\beta$\} be a set of physical characteristics of
the universe, each characteristic perhaps defined in terms of the
values of one or more physical variables at multiple locations and/or
multiple times. Finally, have \{${{\DDD}}_\alpha$\} be all prediction /
observation systems concerning the universe that all scientists might
ever be involved in.

In this example the laws of physics are embodied in $U$. The
implications of those laws for the relationships among the agent
devices \{${{\DDD}}_\alpha$\} and the other characteristics of the universe
\{$\Gamma_\beta$\} is embodied in the reduced form of the
reality. Viewing the universe this way, it is the $u \in U$,
specifying the universe's state for all time, that has ``physical
meaning''. The reduced form instead is a logical implication of the
laws of the universe. In particular, our universe's $u$ picks out the
tuple given by the Cartesian product $[\varprod_\alpha {{\DDD}}_\alpha (u)] \times [\varprod_\beta
\Gamma_\beta(u)]$ from all tuples in the reduced form of the reality.

As an alternative we can view the reduced form of the reality itself as
encapsulating the ``physical meaning'' of the universe. In this
alternative $u$ does not have any physical meaning. It is only the
relationships among the inferences about $u$ that one might want to
make and the devices with which to try to make those inferences that
has physical meaning. One could completely change the space $U$ and
the functions defined over it, but if the associated reduced form of
the reality does not change, then there is no way that the devices in
that reality, when considering the functions in that reality, can tell
that they are now defined over a different $U$.  In this view, the
laws of physics i.e., a choice for the set $U$, are simply a
calculational shortcut for encapsulating patterns in the reduced form
of the reality. It is a particular instantiation of those patterns
that has physical meaning, not some particular element $u \in U$.

See~\cite{tegmark2008mathematical}
for another perspective on the relationship between physical reality and mathematical
structures.
\label{ex:reality}
\end{example}

Given a reality $(U; \{(X_1, Y_1), (X_2, Y_2),
\ldots \})$, we say that a pair of devices in it are
{\bf{pairwise (setup) distinguishable}} if they are distinguishable. 
We say that the reality as a whole is {\bf{mutually (setup)
distinguishable}} iff $\forall \; x_1 \in X_1(U), x_2 \in X_2(U),
\ldots \; \exists \; u
\in U$ s.t. $X_1(u) = x_1, X_2(u) = x_2, \ldots$. 

\begin{proposition}
\label{prop:prop3}

{\bf{i)}} There exist realities $(U; {{\DDD}}_1, {{\DDD}}_2, {{\DDD}}_3)$ where each pair
of devices is pairwise setup distinguishable and ${{\DDD}}_1 > {{\DDD}}_2 > {{\DDD}}_3 > {{\DDD}}_1$. 

{\bf{ii)}} There exists no reality $(U; \{{{\DDD}}_i : i \in {\mathscr{N}}
\subseteq {\mathbb{N}}\})$ where the devices are mutually
distinguishable and for some integer $n$, ${{\DDD}}_1 > {{\DDD}}_2 >
\ldots > {{\DDD}}_n > {{\DDD}}_1$.

{\bf{iii)}} There exists no reality $(U; \{{{\DDD}}_i : i  \in {\mathscr{N}}
\subseteq {\mathbb{N}}\})$ where for some
integer $n$, ${{\DDD}}_1 \gg {{\DDD}}_2 \gg \ldots \gg {{\DDD}}_n \gg {{\DDD}}_1$.
\end{proposition}

There are many ways to view a reality that
contains a countable set of devices $\{{{\DDD}}_i\}$ as a graph, for example by having
each node be a device while the edges between the nodes concern
distinguishability of the associated devices, or concern whether one
weakly infers the other, etc. In particular, given a countable
reality, define an associated directed graph by identifying each
device with a separate node in the graph, and by identifying each
relationship of the form ${{\DDD}}_i \gg {{\DDD}}_j$ with a directed edge going from
node $i$ to node $j$.  We call this the {\bf{strong inference graph}}
of the reality.

Prop.~\ref{prop:whats_inferrable}(ii) means that no reality with $|U| > 3$ can have a universal
device if the reality contains all functions defined over $U$. Suppose that this is not the case,
so that the reality may contain a universal device.
Prop.~\ref{thm:thm_3} means that such a universal device must be a root node
of the strong inference graph of the reality and that there cannot be any other root node. 
In addition, by Prop.~\ref{thm:thm_2}(ii), we
know that every node in a reality's strong inference graph with successor nodes has edges
that lead directly to every one of those successor nodes (whether or not
there is a universal device in the reality). By Prop.~\ref{prop:prop3}(iii) we also
know that a reality's strong inference graph is acyclic. 

Note that even if a device ${{\DDD}}_1$ can strongly infer all other devices
${{\DDD}}_{i > 1}$ in a reality, it may not be able to infer them
$simultaneously$ (strongly or weakly). For example, define a ``composite'' function $\Gamma : u
\rightarrow (Y_2(u), Y_3(u), \ldots)$. Then the fact that ${{\DDD}}_1$ is a
universal device does not mean that $\forall \delta \in {\mathcal{P}}(\Gamma) \;\exists
\; x_1 : Y_1 = \delta(\Gamma)$. See the discussion in~\cite{wolp01} on
``omniscient devices'' for more on this point.

We now define what it means for two devices to operate in an identical
manner:

\begin{definition}
\label{def:def7}
Let $U$ and $\hat{U}$ be two (perhaps
identical) sets. Let ${{\DDD}}_1$ be a device in a reality with domain
$U$. Let $R_1$ be the relation between $X_1$ and $Y_1$ specified by
the reduced form of that reality, i.e., $x_1 R_1 y_1$ iff the pair
$(x_1, y_1)$ occurs in some tuple in the reduced form of the
reality. Similarly let $R_2$ be the relation between $X_2$ and $Y_2$
for some separate device ${{\DDD}}_2$ in the reduced form of a reality having
domain ${\hat{U}}$.

Then we say that ${{\DDD}}_1$ {\bf{mimics}} ${{\DDD}}_2$ iff there is an injection,
$\rho_X : X_2({\hat{U}}) \rightarrow X_1(U)$ and a bijection $\rho_Y :
Y_2({\hat{U}}) \leftrightarrow Y_1(U)$, such that for $\forall x_2,
y_2$, $x_2 R_2 y_2 \Leftrightarrow \rho_X(x_2) R_1 \rho_Y(y_2)$. If
both ${{\DDD}}_1$ mimics ${{\DDD}}_2$ and vice-versa, we say that ${{\DDD}}_1$ and ${{\DDD}}_2$
are {\bf{copies}} of each other.
\label{def:mimic}
\end{definition}

Intuitively, when expressed as devices, two physical systems are
copies if they follow the same inference algorithm with $\rho_X$ and
$\rho_Y$ translating between those systems. As an example, consider the
case where $U = \hat{U}$, and we have a reality over that space that
contains two separate physical computers that are inference
devices, both being used for prediction. If those devices are copies
of each other, then they form the same conclusion for the same value
of their setup function, i.e., they perform the same computation for
the same input. 

The requirement in Def.~\ref{def:mimic} that $\rho_Y$ be surjective
simply reflects the fact that since we're considering devices, $Y_1(U)
= Y_2(U) = {\mathbb{B}}$. Note that because $\rho_X$ in Def.~\ref{def:mimic} need not be surjective,
there can be a device in $U$ that mimics multiple devices in $\hat{U}$. The relation of one device mimicing another
is reflexive and transitive. The relation of two devices being copies
is an equivalence relation.

Say that an inference device $\DDD_2$ is being used for observation and
$\DDD_1$ mimics $\DDD_2$. The fact that $\DDD_1$ mimics $\DDD_2$ does not imply
that $\DDD_1$ can emulate the observation that $\DDD_2$ makes of some
function $\Gamma$. The mimicry property only relates $\DDD_1$ and
$\DDD_2$, with no concern for  relationships with any third
function.
This is why up above we formalized what it means for one device that ``emulates'' another 
in terms of strong inference rather than in terms of mimicry. Indeed, there are some
interesting relationships between what it means for devices to be copies
and what it means for one to strongly infer the other:

\begin{proposition}
\label{prop:prop5}
Let ${{\DDD}}_1$ be a copy of ${{\DDD}}_2$ where both exist in the same reality.

{\bf{i)}} It is possible that ${{\DDD}}_1$ and ${{\DDD}}_2$ are distinguishable and
${{\DDD}}_1 > {{\DDD}}_2$, even for finite $X_1(U), X_2(U)$.

{\bf{ii)}} It is possible that ${{\DDD}}_1 \gg {{\DDD}}_2$, but only if $X_1(U)$ and
$X_2(U)$ are both infinite.
\end{proposition}

\subsection{Philosophical implications}
\label{sec:philo}


Return now to the case where $U$ is a set of laws of physics (i.e.,
the set of all histories consistent with a set of such laws). The
results above provide general restrictions that must
relate any devices in such a universe, regardless of the detailed
nature of the laws of that universe. In particular, these results
would have to be obeyed by all universes in a
multiverse~\cite{smol02,agte05,carr07}.

Accordingly, it is interesting to consider these results from an
informal philosophical perspective. Say we have a device $\DDD$ in a
reality that is distinguishable from the set of all the other devices in the reality. Such a device can be viewed
as having ``free will'', in the limited sense that the way the other devices are set up
does not restrict how $\DDD$ can be set up.  Under this interpretation,
Prop.~\ref{prop:dist_not_infer} means that if two devices both have free will, then they cannot
predict / recall / observe each other with guaranteed complete
accuracy.  A reality can have at most one of its devices that has free
will and can predict / recall / observe / control the other devices in that
reality with guaranteed complete accuracy.{\footnote{There are 
other ways to interpret the vague term ``free will''. For example, Lloyd has
argued that humans have ``free will'' in the sense that under the assumption that they are computationally universal,
then due to the Halting theorem they cannot predict their own future conclusions ahead of
time~\cite{lloyd2012turing}. The fact that an ID cannot even weakly infer itself
has analogous implications that hold under a broader range of assumptions concerning human 
computational capability. For example, this implications hold under
the assumption that humans are \textit{not} computationally universal, or, at the opposite
extreme, under the assumption that they have super-Turing reasoning capability.}

Prop.~\ref{thm:thm_3} then goes further and considers devices that can emulate each
other. It shows that independent of concerns of free will, no two
devices can unerringly emulate each other. (In other words, no reality
can have more than one universal device.) Somewhat tongue in cheek,
taken together, these results could be called a ``monotheism
theorem''.

Prop.~\ref{prop:prop5} tells us that if there is a universal device in some reality, then it must be infinite
(have infinite $X(U)$) if there are other devices in the reality that
are copies of it. Now the time-translation of a physical device is a
copy of that device.{\footnote{Formally, say that the states of some physical system $S$ at
a particular time $t$ and shortly thereafter at $t + \delta$ are
identified as the setup and conclusion values of a device $\DDD$. In
other words, $\DDD$ is given by the functions $(X(u), Y(u))
\triangleq (S(u_t), S(u_{t+\delta}))$. In addition, let $R_S$ be the
relation between $X$ and $Y$ specified by the reduced form of the
reality containing the system. Say that the time-translation of $\DDD$,
given by the two functions $S(u_{t'})$ and $S(u_{t' +
\delta})$, also obeys the relation $R_S$. Then the pair of functions
$(X_2(u), Y_2(u)) \triangleq (S(u_{t'}), S(u_{t' + \delta}))$ is
another device that is copy of $\DDD$. So for example, the same
physical computer at two separate pairs of moments is two separate
devices, devices that are copies of each other, assuming they have the
same set of allowed computations.}} Therefore any physical device that is \emph{ever}
universal must be infinite. In addition, the impossibility of multiple
universal devices in a reality means that if any physical device is
universal, it can only be so at one moment in time. (Its
time-translation cannot be universal.) Again somewhat tongue in cheek,
taken together this second set of results could be called an
``intelligent design theorem''. 

In addition to the questions addressed by the monotheism and
intelligent design theorems, there are many other semi-philosophical
questions one can ask of the form ``Can there be a reality with the
following properties?''. By formulating such questions in terms
of reduced realities, they can often be
reduced to constraint satisfaction problems,
potentially involving infinite-dimensional spaces. In this sense,
many of the questions that have long animated philosophy can be 
formulated as constraint satisfaction problems.

\section{Physical knowledge}
\label{sec:phys_know}

Say that colloquially speaking you ``know'' the sky's color is currently blue, 
so long as $u$ is in some
subset $W$ of all histories. (The reason we consider subsets $W$
is that you cannot know that the sky's color is blue
in \textit{all} histories, since in some histories it will \textit{not}
be blue.) How can we
formalize this colloquial notion? Well, one thing it means if you ``know the sky's color
is blue'' for any $u \in W$ is that for such $u$'s you can ask yourself ``Is the sky green?'' and
answer 'no', ask yourself ``Is the sky red?'' and answer 'no', ask
yourself ``Is the sky blue'' and answer 'yes', etc., and always be
correct in your answer. So to ``know''
something implies you can weakly infer it.  Intuitively speaking, weak
inference formalizes an aspect of the semantic content of ``knowledge''.

To properly formalize knowledge of the sky's color however, we need to use more structure than
is just provided by weak inference of the sky's color. 
The problem is that it is possible that $(X, Y) > \Gamma$ even if for each $\gamma \in\Gamma(U)$,
the associated $x$ that causes $Y(u) = \delta(\gamma, \Gamma(u))$ always results in $Y(u) = -1$.{\footnote{Note
that there must be \emph{some} $x$ that allows $Y(u) = 1$, since $|Y(U)| = 2$. However
it may be that none of those specific $x$'s that are involved in the ID's inferring $\Gamma$ have that property.}}
Loosely speaking, $(X, Y)$ can infer the sky's color by always setting itself
up so that it (correctly) answers
that the sky does not have a given color $c$, so long as it can do that 
for any given color $c$.{\footnote{This characteristic  of weak
inference is an example of how flexible and unrestrictive the definition of weak
inference is, mentioned above. This particular flexibility is most reasonable for the
inference process of control, where typically $x$ directly influences the value of $\Gamma$,
and to a somewhat lesser degree for the inference process of observation.}} So to say that 
$(X, Y)$ knows $\Gamma = \gamma$ over (all $u$ in) $W$, it makes sense not just to require that
$(X, Y) > \Gamma$, but also that for all $\gamma \in \Gamma(U)$, 
there exists some $u \in W$ such that both $X(u)$ is a setup
value that arises for the question, ``Does $\Gamma(u) = \gamma$?'', and that $Y(u) = 1$, i.e.,
that the device answers `yes'.

Similarly, it would be problematic to say that the device $(X, Y)$ ``knows'' the sky's color
if $(X, Y)$ can infer the sky's color by always setting itself
up so that it (correctly) answers
that the sky \emph{does} have a given color $c$, so long as it can do that 
for any given color $c$. This suggests we want to also add the requirement 
that  for all $\gamma \in \Gamma(U)$, there exists some $u \in W$ such that $X(u)$ is a setup
value that arises for the question, ``Does $\Gamma(u) \ne \gamma$?'', and that $Y(u) = -1$, i.e.,
that the device answers `no'.

To model knowledge in this sense, not just inference, we need to guarantee that
there is some color $c$ such that whenever the history $u$ is in
some set $W$, for the question, ``Is the sky's color \emph{c}?'', the
inference device will answer `yes', and be correct. In other words,
you don't ``know'' the sky's color whenever $u \in W$ if you can only ever
say what color it is \emph{not} whenever $u \in W$. For it to be the case
that whenever $u \in W$ you know that the sky's color is \emph{c}, at
a minimum, it must be that you can correctly answer ``yes, the sky's
color is \emph{c}'', for some such $u \in W$. Nonetheless, we also want to
guarantee that there is at least one $u \in W$ at which we correctly answer
``no, $c'$ is not the sky's color'' for \textit{some} color $c'$, which may
be the same as $c$ or different.
%

\subsection{Formal definition of physical knowledge}

We can formalize this strengthened version of inference as follows:

\begin{definition}
Consider an inference device $(X, Y)$ defined over $U$, a function $\Gamma$ defined
over $U$, a $\gamma \in \Gamma(U)$, and a subset $W \subseteq U$. We say that
``\textbf{$(X, Y)$ \textbf{(physically) knows}
$\Gamma = \gamma$} over $W$'' iff $\exists$ $\xi :
\Gamma(U) \rightarrow \overline{X}$ such that
\begin{enumerate}[i)]
\item  $\forall \gamma' \in \Gamma(U), u \in  \xi(\gamma') \Rightarrow \delta_{\gamma'}(\Gamma(u)) = Y(u)$,
\item $\varnothing \;\ne\; \xi(\gamma) \cap W \;\subseteq\; Y^{-1}(1)$.
\item For all $\gamma' \ne \gamma$, $\varnothing \;\ne\; \xi(\gamma') \cap W \;\subseteq\; Y^{-1}(-1);$
\end{enumerate}
\label{def:knows}
\end{definition}
\noindent (Recall that $\overline{X}$ is the partition of $U$ induced by $X$.) 
When I want to specify the precise function $\xi$ used in Def.~\ref{def:knows},
I will say that ``by using $\xi$, $(X, Y)$ knows that $\Gamma = \gamma$ over $W$". 

By Def.~\ref{def:knows}(i), if $(X, Y)$ physically knows $\Gamma = \gamma$ over $W$, then $(X, Y)$ weakly infers $\Gamma$.
So $(X, Y)$ is always correct in its inference --- even if $u \not \in W$.
We impose this requirement for all of $U$, not just $W$,
because the agent using the device does not have any \emph{a priori} reason to expect that $u \in W$.
So it does them no good to be able to set up a device
that will correctly say whether some function has a certain value --- but only if 
the condition $u \in W \subset U$ holds, a condition they cannot detect.

Def.~\ref{def:knows}(ii) and  Def.~\ref{def:knows}(iii) are the extra conditions
beyond just weak inference, forcing the ID
to answer 'yes' at least once, and to answer 'no' at least once. Neither of those conditions
depend on the precise form of the function $\Gamma(u)$, only its image, $\Gamma(U)$
(which specifies the domain of $\xi$).
It's also worth noting that most of the analysis below does not invoke Def.~\ref{def:knows}(iii).
The motivation for including that condition anyway will arise below, when
we demonstrate that physical knowledge need not imply logical omniscience;
this demonstration is more consequential if it applies even when Def.~\ref{def:knows}(iii) holds.


The following properties are immediate:
\begin{lemma}
Let $(X, Y)$ be a device defined over $U$, $\Gamma$ a function over $U$, and
$W$ a subset of $U$. Say that by using $\xi$,
$(X, Y)$ knows that $\Gamma = \gamma$ over $W$. It follows that:
\begin{enumerate}[i)]
\item $\Gamma(u) = \gamma \; \forall u \in \xi(\gamma) \cap W$;
\item If $W$ refines $\Gamma$, then $\Gamma(W) = \gamma$.
\end{enumerate}
\label{lemma:elem}
\end{lemma}
\begin{proof}

To prove the first claim, note from Def.~\ref{def:knows}(ii)
that for all $u \in \xi(\gamma) \cap W$, $Y(u) = 1$. By
Def.~\ref{def:knows}(i), this means that at all such $u$, $\Gamma(u) = \gamma$,
completing the proof. Given this, if in addition
$W$ refines $\Gamma$ (so that $\Gamma(u)$ has the same value across
all $W$), then it must be that  $\Gamma(u) = \gamma$ for all
$u \in W$. (Similar arguments for $Y(u) = -1$
follow by using Def.~\ref{def:knows}(iii).) This establishes the second claim.
\end{proof}


Note that the definition of physical knowledge does not require that $\xi(\gamma) \subseteq Y^{-1}(1)$,
but only that $\xi(\gamma) \cap W \subseteq Y^{-1}(1)$ (and similarly for $Y^{-1}(-1)$). The simple fact that 
$\bar{x} \in \xi(\gamma)$ \emph{and nothing more} does
not imply that the device must answer `yes' if $X(u) = \xi(\gamma)$.
Furthermore, there may be more than one $\xi(.)$ with which the ID can ``know $\Gamma = \gamma$ over $W$''.
There may even be some other $\xi(.)$
that can be used to instead know $\Gamma = \gamma' \ne \gamma$ over $W$. This illustrates
that physical knowledge does not require that $\Gamma$ have the same value over all of $\xi(\Gamma)$.
This flexibility means that physical knowledge
includes knowledge that occurs by observation of the value of $\Gamma$, 
just like inference does.

A related point is that we do not require that $W$ refine $\Gamma$ to have a device
know that $\Gamma = \gamma$. This freedom allows the device to know
that $\Gamma = \gamma$ over $W$ even if the value of $\Gamma$ depends
on the value of $X$, the question the device is asking. In other words,
it is possible that the device both knows that $\Gamma = \gamma$ over $W$ 
and knows that $\Gamma = \gamma'$ over $W$  for some $\gamma' \ne \gamma$.
In this sense, the definition of physical knowledge is extremely non-restrictive.

This lack of restriction means that physical knowledge allows for
``quantum-mechanical-style'' coupling of an observation device and
the system being observed. More generally, it allows $(X, Y)$
to be a device that controls the property $\Gamma$ of the system
being observed. Typically though, when we are interested in knowledge in the sense
of accurate prediction or observation that does not affect
the system being predicted / observed, $W$ will refine $\Gamma$. 

The following example illustrates Def.~\ref{def:knows} in more detail:
\begin{example}
Say that the sky above Greenwich, UK at time $t$ 
is \{blue, cloud-free, with the sun less than 15 degrees above
the horizon\}. Furthermore, say that at some
time $t'$, Bob knows that the sky above Greenwich, UK at time $t$ is blue.
(It does not matter whether $t' = t$.)  To formalize this knowledge in terms 
of Def.~\ref{def:knows}, let $U$ be
the set of all histories in which both Bob and Greenwich, UK exist, and
where in addition the following conditions hold:
\begin{enumerate}[i)]
\item There is a partition
${\mathcal{C}}$ of all possible distributions of the intensity of light in optical wavelengths. For
example, one element of that partition is `green', one is `red', and one is the color 'blue';
\item Bob asks himself at $t'$, ``Is $c$ the color of the 
sky above Greenwich at $t$?'', for some color $c \in {\mathcal{C}}$;
\item Bob answers that question at that time $t'$ to the best of his abilities, with either a
'yes' or a 'no'. 
\end{enumerate}
Define $\Gamma(u)$ as the map taking each $u \in U$ to the associated
element of ${\mathcal{C}}$ that characterizes the color of the sky above Greenwich
at $t$. Define $X(u)$  as the map taking each $u \in U$ to the associated
color $c$ where at $t'$ Bob is asking himself the
question, ``Does the sky's color at Greenwich at $t$ equal
$c$?''. Assume that the image of $X$ is all ${\mathcal{C}}$. Next, let $Y(u)$
specify the binary answer in Bob's mind at $t'$. Finally, let
$W \subseteq U$ be the set of all histories $u$ such that
the sky above Greenwich at time $t$
is \{blue, cloud-free, with the sun less than 15 degrees above
the horizon\}, and assume $u \in W$.

Given all this, ``at $t'$, Bob knows the sky is blue above Greenwich at
time $t$, when the sky above Greenwich at time $t$
is \{blue, cloud-free, with the sun less than 15 degrees above
the horizon\}'', in the sense of Def.~\ref{def:knows} if three conditions
are met:
\begin{enumerate}[i)]
\item There is a $u$ for which ``the sky above Greenwich at time $t$
is \{blue, cloud-free, with the sun less than 15 degrees above
the horizon\}'' and for which $X(u)$ specifies the question, ``Is the sky
blue above Greenwich at $t$?'', i.e.  for which Bob is asking himself that
question at $t'$; 
\item For the $u$ in (i), $Y(u)$ specifies that Bob answers 'yes' at $t'$; 
%
\item If the sky above Greenwich at time $t$
were still \{blue, cloud-free, with the sun less than 15 degrees above
the horizon\}, but instead Bob were able to ask himself at $t'$ any question 
of the form ``Does the sky's color at Greenwich at $t$ equal $c'$?''
where $c' \ne$ blue (i.e., if the history were some different $u' \in W$
where $X(u') \ne X(u)$) and did so, Bob would answer 'no' at $t'$ (i.e., $Y(u)$ would equal $-1$).
\end{enumerate}

This is a very simple example. In particular, in some situations for
Bob to know at $t'$ that the sky is blue at Greenwich at $t$, Bob will
need to configure an apparatus to have a particular state at a
particular time (e.g., he may need to configure an automatic camera to
photograph the sky above Greenwich at $t$). In such situations, $X$
not only specifies the question that Bob asks himself at $t'$ but
also specifies how Bob configures the apparatus. See~\cite{wolp08b}.

Note that the set $U$ in this example might be a proper subset of the
set of all histories that are consistent with the laws of
physics. 
This is just an example of the
fact that the definition of weak inference in general implicitly
specifies a set $U$ that is a subset of the set of all physically possible
histories.  Indeed, there might very well be histories that are
consistent with the laws of physics in which Bob does not exist, or
perhaps even Greenwich does not exist. Clearly we cannot speak of
whether Bob does or does not ``know the sky is blue'' in any such
histories.{\footnote{This need to restrict the universe of discourse
to a subset of all physically possible histories holds no matter how
we formalize ``knowledge''. It has nothing to do with formalizing
knowledge using inference devices.}}

Finally, note that Bob could conceivably also know at time $t$ that the
sky above Greenwich, UK at time $t$ is cloud-free, or that the sun in
that sky is less than 15 degrees above the horizon. Any such
alternative knowledge would require posing a question to Bob about a
different subject. Therefore it would require a different value of
$X(u)$, and therefore a different $u$. However $W$, the state of
the sky above Greenwich at $t$, doesn't vary with the question that Bob asks
concerning that sky. This means
that $W$ contains each of those different $u$'s that result in different values of $X(u)$.
Ultimately, it is to allow this possibility of multiple questions all concerning the
same state of the sky that the definition of physical knowledge involves
sets $W$. 
\label{ex:1}
\end{example}

%

\subsection{Epistemic logic based on physical knowledge}
\label{sec:boole}


Inference knowledge obeys many of the usual properties considered
in theories of  logic. Here I illustrate this by
presenting some of those properties.

For the rest of this subsection assume that any space $U$ we consider is countable. 
I will consider Boolean-valued functions, i.e., functions $\Gamma$ such that $\Gamma(U) = \{-1, 1\}$.
It will also be convenient to identify 
each such binary-valued function $\Gamma$ over $U$ with the associated set $\Gamma^{-1}(1)$. 
Using this identification, any so-called \emph{concrete} Boolean algebra
over subsets of $U$ defines a Boolean algebra over an associated set of binary-valued functions, which
we call the \textbf{function Boolean algebra}. (Equivalently, a function
Boolean algebra over $U$ is a Boolean algebra specified by a set of
bit strings indexed by elements of $U$.)

Moreover, we can always express an arbitrary Boolean algebra
as a concrete Boolean algebra~\cite{wiki_boolean_algebra}.
Accordingly, given any Boolean algebra with propositions $\Phi$, we can identify any
specific proposition $\phi \in \Phi$ as a subset of $U$ in the concrete Boolean algebra, and then 
identify that element of the concrete Boolean algebra with an element of the function Boolean algebra. So we can identify
any proposition $\phi$ with a specific binary function, which we write as $\Gamma_\phi$ (with the 
set $\Phi$ usually implicit).

We now use the function Boolean algebra to define the standard shorthands of propositional logic
for binary-valued functions. For example, for any two binary-valued functions $\Gamma_1$ and $\Gamma_2$, 
their logical AND is 
\ba
\Gamma_1(u) \wedge \Gamma_2(u) = \{u \in U : \Gamma_1(u) = \Gamma_2(u) = 1\}
\label{eq:AND}
\ea
and the logical NOT is given by
\ba
\neg \Gamma_1(u)  = \{u \in U : \Gamma_1(u) = -1\}
\ea

This allows us to apply the usual axioms of Boolean algebra to binary-valued functions.
We can also adopt the usual abbreviations from Boolean algebra, e.g.,
\ba
\Gamma_1 \vee {\Gamma_2} &=& \neg (\neg \Gamma_1 \wedge \neg {\Gamma_2}); \\
\Gamma_1 \equiv \Gamma_2 &=& (\Gamma_1 \wedge \Gamma_2) \vee (\neg \Gamma_1 \wedge \neg {\Gamma_2}); \\
\label{eq:6}
\Gamma_1 \Rightarrow {\Gamma_2} &=& \neg\Gamma_1 \vee {\Gamma_2}; \\
\label{eq:7}
\Gamma_1 \Leftrightarrow {\Gamma_2} &=& (\Gamma_1 \Rightarrow {\Gamma_2}) \wedge ({\Gamma_2} \Rightarrow \Gamma_1)
\ea
I extend this terminology to cases where we are considering subsets $W \subseteq U$ in the obvious way,
e.g., $\Gamma_1 \Rightarrow {\Gamma_2}$ has the value
`true' over $W$ iff $\neg\Gamma_1(u) \vee {\Gamma_2}(u)$ has the value `true' for all
$u \in W$.

Similarly, though it is not used in this paper, as is conventional we can 
take the function TRUE$(u)$ to be an
abbreviation for some fixed propositional tautology such as $(\Gamma_1 \vee \neg \Gamma_1)(u)$ (i.e.,
the function that equals 1 for all $u \in U$) and
FALSE to be an abbreviation for the function $\neg$ TRUE.{\footnote{
Note that these two function each violate the requirement of the usual formulation of
inference devices  that every function take on at least two
values. So in particular, in order to consider whether a device can
weakly infer such a function, we would need to weaken the definition of weak inference of a function
to allow a single value in the image of that function.}} In keeping with this, 
I will sometimes use the term 'true' to mean the value $1$,
and use the term `false' to mean the value $-1$.

In many conventional types of epistemic logic, in particular Kripke structures, knowledge is
defined in such a way that is impossible for an agent to know two contradictory things.
However as discussed above, to allow physical knowledge to capture the case where
the agent ``knows a function has a specific value'' due to the (physical) 
fact that they control the value of that
function, we are careful to define terms so
that an agent can physically know contradictory things. Specifically,
this occurs if the function that they physically know, $\Gamma$,
takes on more than one value across the set under consideration, $W$, and by appropriate
choice of ($\xi$ and therefore) $\bar{x}$, the agent can cause different probes of $\Gamma(u)$
to have the value 1.  (Note that in this case $\Gamma$ is \emph{not} refined
by $W$.) 

Because of this, certain epistemic properties that are automatically satisfied in
conventional types of epistemic logic must be carefully derived in analysis of physical knowledge.
The following proposition presents one of these properties. For pedagogical clarity, 
in the remainder of this subsection, ``is true'' is shorthand for ``is
true over $W$'' and similarly ``is false'' is shorthand for ``is false
over $W$''.

\begin{proposition}
Let $\Gamma$
be any binary-valued function over $U$, $\DDD = (X, Y)$ any device over $U$,
and $W$ some (implicit) subset of $U$. Then\
$\DDD$ knows that $\Gamma$ is false iff $\DDD$ knows that $\neg \Gamma$ is true.
\label{prop:first_impl} 
\end{proposition}

\begin{proof}
I prove the forward direction; the inverse follows the same way.
Let $\xi_{\Gamma}$ be the operator establishing that 
 ``$\DDD$ knows that ${\Gamma}$ is false".
So $\xi_{\Gamma}(-1) \cap W \subseteq Y^{-1}(1)$. Define $\xi_{\neg {\Gamma}}(\gamma) = \xi_{{\Gamma}}(-\gamma)$
for all $\gamma \in \B$ (i.e., for all $\gamma$ in the codomains of both ${\Gamma}$ and $\neg {\Gamma}$).
It follows that $\xi_{\neg {\Gamma}}(1) \cap W \subseteq Y^{-1}(1)$. 
This establishes that if the condition Def.~\ref{def:knows}(ii)
hold for ``$\DDD$ knows that ${\Gamma}$ is false (over $W$)", then it must also hold
for ``$\DDD$ knows that $\neg {\Gamma}$ is true". 
Property (iii) is established using an analogous argument.
Finally if property (i) in Def.~\ref{def:knows}
holds  for ``$\DDD$ knows that $ {\Gamma}$ is false" using $\xi_{ \Gamma}$, it must also hold
for ``$\DDD$ knows that $\neg {\Gamma}$ is true" using $\xi_{\neg {\Gamma}}$. This establishes the
(forward direction of the) claim in full.


\end{proof}

Note that applying Prop.~\ref{prop:first_impl}(i) with a new function $\Gamma' := -\Gamma$ establishes that a device
$\DDD$ knows that a function $\Gamma'$ is true iff $\DDD$ knows that $\neg \Gamma'$ is false.

In contrast to these results, if the function under consideration is refined by $W$, then the agent cannot
know contradictory things concerning the value of that function. 
We start our discussion of this kind of situation with an immediate corollary of Lemma~\ref{lemma:elem}(ii), which 
intuitively says that a device cannot physically know something if that thing is false:
\begin{corollary}
Suppose that $W$ refines $\Gamma$. Then 
if $\DDD$ knows that $\Gamma$ is true (over $W$), $\Gamma$ is true. 
\label{corr:trivial}
\end{corollary}
(Similarly, if $W$ refines $\Gamma$ and $\DDD$ knows that $\Gamma$ is false, $\Gamma$ is false.) 
Note that Coroll.~\ref{corr:trivial} tells us that if $W$ refines ${\Gamma}$, then 
it is possible that $\DDD$ knows that ${\Gamma}$ is true, or that
$\DDD$ knows that ${\Gamma}$ is false --- but not both. So if $W$ refines $\Gamma$, 
we have the usual property that $\DDD$ cannot know two contradictory things.

Since binary-valued
functions obey the rules of propositional logic, Coroll.~\ref{corr:trivial} means that if $W$ refines $\Gamma$ as well as $\Gamma'$,
and $\DDD$ both knows that $\Gamma$ is true and that $\Gamma'$ is true,
it follows that $\Gamma \wedge \Gamma'$ is true. 
This immediately establishes many
properties of physical knowledge --- in particular if $\Gamma'$ involves the logical
$\Rightarrow$ operator defined in Eq.~\eqref{eq:6} --- including the following:

\begin{corollary}
Let $\Gamma_1$, ${\Gamma_2}$ and ${\Gamma_3}$ be any binary-valued functions over $U$, 
and $\DDD = (X, Y)$ any device over $U$.
\begin{enumerate}[i)]

\item Say that  $W$ refines $\Gamma_1$.
Then if $\DDD$ knows that $\Gamma_1$ is true, and $\Gamma_1 \Rightarrow {\Gamma_2}$ is
true, it follows that ${\Gamma_2}$ is true.

\item Say that  $W$ refines $\Gamma_1$ as well as $\Gamma_1 \Rightarrow {\Gamma_2}$.
Then if $\DDD$ knows that $\Gamma_1$ is true, and knows that $\Gamma_1 \Rightarrow {\Gamma_2}$ is
true, it follows that ${\Gamma_2}$ is true.

\item Say that  $W$ refines  $\Gamma_1 \Rightarrow {\Gamma_2}$ as well as
${\Gamma_2} \Rightarrow {\Gamma_3}$. Then if $\DDD$ both knows that 
$\Gamma_1 \Rightarrow {\Gamma_2}$ is true and knows that
${\Gamma_2} \Rightarrow {\Gamma_3}$ is true, it follows that $\Gamma_1 \Rightarrow {\Gamma_3}$ is true.


\item Say that we have a set of binary-valued functions  $\{{\Gamma}_i : i = 1, \ldots N\}$
and that  $W$ refines $\Gamma_1$ as well as $\Gamma_i \Rightarrow \Gamma_{i+1}$ for all $i \in \{1, \ldots, N-1\}$. Then if 
$\DDD$ knows that $\Gamma_1$ is true and knows that $\Gamma_i \Rightarrow \Gamma_{i+1}$ is true 
for all $i \in \{1, \ldots, N-1\}$, it follows that $\Gamma_i$ is true for all $i \in \{1, \ldots, N\}$.

\end{enumerate}
\label{prop:impl} 
\end{corollary}
\noindent (In interpreting these results, the reader should remember to insert ``over $W$'' after every statement
about whether a function is true or false.)

In general, the properties described in  Coroll.~\ref{prop:impl} do not hold without the 
conditions that certain functions are refined by $W$. So for the most part, they need not
hold for the case where an agent knows the value of a function because they control its value.
Note though that Coroll.~\ref{prop:impl}(ii) does not require that $W$ refine ${\Gamma_2}$. 
Similarly Coroll.~\ref{prop:impl}(iii) does not require that $W$ refine either $\Gamma_1$,
$\Gamma_2$,  $\Gamma_3$, or  ${\Gamma_1} \Rightarrow {\Gamma_3}$.

%
%

We can weaken the last two claims in Coroll.~\ref{prop:impl}:
\begin{corollary}
Let ${\Gamma_1}$ and ${\Gamma_2}$ be any binary-valued functions over $U$, $\DDD$ any device over $U$,
and $W$ any (implicit) subset of $U$.
\begin{enumerate}[i)]
\item Say that $W$ refines ${\Gamma_1}$ and refines ${\Gamma_1} \Rightarrow {\Gamma_2}$.
Then if either $\DDD$ knows that ${\Gamma_1}$ is true and 
${\Gamma_1} \Rightarrow {\Gamma_2}$ is true, or  ${\Gamma_1}$ is true and 
$\DDD$ knows that
${\Gamma_1} \Rightarrow {\Gamma_2}$ is true, 
it follows that ${\Gamma_2}$ is true.
\item Say that  $W$ refines ${\Gamma_1} \Rightarrow {\Gamma_2}$ and refines
${\Gamma_2} \Rightarrow {\Gamma_3}$. Then if either $\DDD$ both knows that 
${\Gamma_1} \Rightarrow {\Gamma_2}$ is true and
${\Gamma_2} \Rightarrow {\Gamma_3}$ is true, or
${\Gamma_1} \Rightarrow {\Gamma_2}$ is true and  $\DDD$ knows that 
${\Gamma_1} \Rightarrow {\Gamma_3}$ is true
it follows that ${\Gamma_2} \Rightarrow {\Gamma_3}$ is true.
\end{enumerate}
\end{corollary}

\subsection{Impossibility results concerning physical knowledge}

There are major restrictions on physical knowledge. The first such restriction
follows from the first demon theorem.

\begin{corollary}
For any device $\DDD$, there exists a function $\Gamma$ over $U$ 
such that for no $W \subseteq U, \gamma \in \Gamma(U)$ does $\DDD$ know
$\Gamma = \gamma$ over $W$.
\label{coroll:5}
\end{corollary}

The second major restriction follows from the second  demon theorem.
\begin{corollary}
Let $(X_1, Y_1)$ and $(X_{-1}, Y_{-1})$ be two distinguishable devices. Then for at least one
of the two devices $i \in \{-1, 1\}$, there is no pair $(W \subseteq U, y_{-i} \in Y_{-i}(U))$ such that
$(X_i, Y_i)$ knows that $Y_{-i} = y_{-i}$ over $W$.
\label{coroll:6}
\end{corollary}

Similarly, Coroll.~\ref{coroll:1} provides another restriction on
physical knowledge:
\begin{corollary}
Consider a pair of devices $\DDD = (X, Y)$ and $\DDD ' = (X', Y')$ that are both distinguishable
from one another and whose conclusion functions are inequivalent. Say that there
is a $W \subseteq  U$ that refines $Y$
such that $\DDD'$ knows that $Y = Y(W)$ over $W$. Then there are at least three 
inequivalent surjective binary functions $\Gamma$ such that  there is no $W'$ with
the following two properties: $W'$ refines $\Gamma$, 
and $\DDD$ know that $\Gamma = \Gamma(W')$ when $W'$.
\end{corollary}
In other words, if $\DDD '$ knows the value of $\DDD$'s conclusion function over $W$,
then there are at least three separate functions that $\DDD$ never knows, no
matter what the subset of $U$ we are in.

\subsection{Physical knowledge and the first three rules of \emph{S5}}

\emph{S5} is a set of five rules obeyed by many epistemic logics, including Kripke structures.
The ``knowledge axiom'' of \emph{S5} says that if an agent knows a Boolean proposition $\phi$, then $\phi$ must be true.
We can use the map from propositions to binary functions (discussed just before Eq.~\eqref{eq:AND})
to formulate a physical knowledge version of this axiom. 
Coroll.~\ref{corr:trivial} confirms that the physical knowledge analog of the knowledge axiom holds (assuming 
we only consider sets $W$ that refine $\Gamma_\phi$).

In addition to the knowledge axiom,
\emph{S5} also includes the ``knowledge generalization rule'', which says that if proposition $\phi$ is true
in all possible states of the world (i.e., if $\phi$ is necessarily true rather than contingently true), then the agent knows $\phi$. 
An analogous rule in terms of physical knowledge might be that if $W$ refines ${\Gamma_1}$, and ${\Gamma_1}$ is true over $W$,
then the agent physically knows ${\Gamma_1}$ is true over $W$. However this rule need not hold; an agent $(X, Y)$ may not
be able to weakly infer ${\Gamma_1}$, whether or not $\Gamma_1$ is true over $W$.
(And even if they can infer ${\Gamma_1}$, it may be that $Y^{-1}(1) \cap W = \varnothing$.)

The ``distribution axiom" of \emph{S5} says that if an agent both knows proposition $\phi_1$ and knows $\phi_1 \Rightarrow \phi_2$,
then $\phi_2$ is true \emph{and they know this}. 
In contrast, Coroll.~\ref{prop:impl}(ii) only establishes that if an agent both physically knows that ${\Gamma_1}$
is true over $W$ and knows that ${\Gamma_1} \Rightarrow {\Gamma_2}$ is true over $W$, then ${\Gamma_2}$
is also true over $W$. However it is easy to construct examples where the
conditions for Coroll.~\ref{prop:impl}(ii) hold but the agent does not know
$\Gamma_2$ is true. Ultimately, the reason
for this difference between physical knowledge and Kripke structures
is due to the requirement of physical knowledge of $\Gamma$ that the
device weakly infer $\Gamma$ --- a requirement that involves considering counterfactual scenarios,
something not done in conventional epistemic logics.

For what are ultimately the same reasons, the conditions
for Coroll.~\ref{prop:impl}(i) may hold even if the agent does not physically know that $\Gamma_2$ is true. 
This is illustrated in the following example:

\begin{example}
For purposes of this example, fix some particular location and time.
Suppose that the (binary-valued) function $\Gamma_1(u)$ is
defined by whether the temperature in $u$ is / isn't ten degrees celsius at that location and time.
Have $\Gamma_2(u)$ be whether the temperature in 
$u$ is / isn't above freezing. Presume $U$ is large enough that there are $u \in U$ that satisfy each of the three possible values
of $(\Gamma_1(u), \Gamma_2(u))$.
Finally, have $W$ be the set of all $u$ at which the temperature
is ten degrees. Note that $(\Gamma_1 \Rightarrow \Gamma_2)(u)$ is always true. (This means it is
a ``valid'' statement, in the language of epistemic logic.) So in particular it is true for all $u \in W$.

To ground thinking suppose that the ID is a thermometer that outputs a 1 or -1,
depending on the temperature. $x$ is the value of the temperature that the thermometer
is checking. By Def.~\ref{def:knows}(i), if the agent physically knows $\Gamma_1 = 1$ over $W$, then there
is some setup function $\xi$ they can use to configure their thermometer to correctly give a 1 if 
the temperature is ten degrees, and to correctly give a -1 otherwise. In other words, such
physical knowledge requires that when answering the question,
``Is the temperature ten degrees?'' (i.e., does $\Gamma_1(u) = 1$), the agent will
set $X$ to be in the state $\xi(1)$, and they will be guaranteed that the associated conclusion $Y(u)$
will equal $\delta_1(u)$ for any $u \in \xi(1)$ \underline{whether or not $u \in W$}.

In general, depending on whether that $\xi$ obeys
$\xi(1) \subseteq W$, this phenomenon may mean that the ID gives the correct answer for some $u$ in which the temperature is
not ten degrees. This is key: the ID
is set up with the same $x$ value regardless of whether $u \in W$. \underline{After} the thermometer is set up this way,
\underline{then} the agent learns whether the temperature equals ten degrees. If (after having been set up 
with the $x$ value corresponding to ten degrees) the ID tells the agent that the temperature
is indeed ten degrees (since $u \in W$), {at that point the} agent has physical knowledge that the temperature
is ten degrees. But not before.

In particular, consider the situation where $\xi(1)$ is big enough to contain both a $u$ where the temperature is five degrees, and 
one where it is negative five degrees (in addition to containing one where it is ten degrees). 
Since we require that $Y(u) = \delta_1(\Gamma_1(u))$ for all $u \in \xi(1)$, this means that $Y(u) = -1$ 
for both of those $u$'s. Note though that $\Gamma_2(u)$ has different values for those two temperatures. 
This means that the agent cannot use this same $\xi$ that allows them to 
weakly infer $\Gamma_1$ to also weakly infer $\Gamma_2$. (This is how the intuitive
notion of regular implication would work too; if all I know is that the temperature is not 10 degrees, then I don't know whether it is -5 or 5.)
Moreover, in general, there may not be any alternative to this $\xi$ that the 
agent can use with that thermometer to weakly infer $\Gamma_2$, i.e., weakly 
infer whether the temperature is above freezing or not. In this
situation, the agent cannot weakly infer $\Gamma_2$. (Recall again the key point that in 
the definition of physical knowledge we require that weak inference holds for all $u \in U$,
not just the $u \in W$.)

So in this situation, the agent does not physically know that $\Gamma_2 = 1$ for $u \in W$.
Loosely speaking, there are thermometers that can be used to \underline{always} tell us
correctly whether the temperature is (not) ten degrees  -- both when it is and when
it isn't ten degrees. However some such thermometers cannot be used to \underline{always} tell us
correctly whether the temperature is (not) above freezing (both when it is and when it isn't above freezing). 

Given such a thermometer, for the particular situation where the thermometer is set up to detect whether the temperature is
ten degrees, and in addition it answers 'yes', you and I can use our reasoning ability to
realize that the temperature must above above zero. But the ability to
use such reasoning to come to a conclusion is not the same thing as physical knowledge of that conclusion.

In the rest of this example I establish this argument in a fully formal manner,
making the simplifying assumption that $W$ refines $\Gamma_1$, as in Coroll.~\ref{prop:impl}(i).
First, note that by Lemma~\ref{lemma:elem}(ii), since the ID physically know that $\Gamma_1 = 1$
when $W$, $\Gamma_1(u)$ is true throughout $W$. This in turn means that $\Gamma_2$ is true throughout
$W$ (since $\Gamma_1 \Rightarrow \Gamma_2$). 

Let $\xi$ be the function that establishes that condition Def.~\ref{def:knows}(ii)
holds for $\Gamma_1$. We can use that same function to establish that condition Def.~\ref{def:knows}(ii)
holds for $\Gamma_1 \Rightarrow \Gamma_2$ and for $\Gamma_2$. 
Similar arguments hold for condition Def.~\ref{def:knows}(iii). So $\DDD$ 
meets conditions (ii) and (iii) for having physical knowledge of
both $\Gamma_1 \Rightarrow \Gamma_2$ and $\Gamma_2$. So to complete our analysis
of whether $\DDD$ has physical knowledge that $\Gamma_2 = 1$,
we must consider whether condition Def.~\ref{def:knows}(i) for $\Gamma_2$. We do this by considering two cases, one
in which $\DDD$ does physically know $\Gamma_2 = 1$ when $W$, and one where it does not:

\begin{enumerate}
\item First, assume $\xi(\gamma) \subseteq W$ both
for $\gamma = 1$ and $\gamma = -1$. Since Def.~\ref{def:knows}(i) holds for $\Gamma_1$ for that
$\xi$, and since $\Gamma_1(u) = \Gamma_2(u)$ throughout $W$, it is immediate that Def.~\ref{def:knows}(i)
also holds for that $\xi$. So $\DDD$ knows that $\Gamma_2$ is true over $W$, under our assumption. 

Next, plug 
the fact that $\Gamma_1(u) = \Gamma_2(u) = 1$ for all $u \in W$ into the definition of $\Gamma_1 \Rightarrow \Gamma_2$
to see that $\delta_1((\Gamma_1 \Rightarrow \Gamma_2)(u)) = 1$ for all $u \in \xi(1) \subset W$. So
$\delta_1((\Gamma_1 \Rightarrow \Gamma_2)(u)) = Y(u)$ for all $u \in \xi(1)$. This establishes that 
Def.~\ref{def:knows}(i) holds for the function $\Gamma_1 \Rightarrow \Gamma_2$ for the case of $\gamma' = 1$.
For the remaining case of $\gamma' = -1$, note that for all $u \in \xi(-1)$, again 
$\Gamma_1(u) = \Gamma_2(u) = 1$. So $\delta_{-1}((\Gamma_1 \Rightarrow \Gamma_2)(u)) = -1$. Since
$Y(u) = -1$ throughout $\xi(-1)$, this establishes that Def.~\ref{def:knows}(i) 
also holds for the function $\Gamma_1 \Rightarrow \Gamma_2$ for the case of $\gamma' = -1$. Accordingly,
under our assumption that the support of both $\xi(1)$ and of $\xi(-1)$ is restricted to $W$, $\DDD$ knows that
$\Gamma_1 \Rightarrow \Gamma_2$ over $W$.

\item In many situations however, even though $W$ refines $\Gamma_1$, it will \emph{not} be the case
that the associated function $\xi$ that establishes that Def.~\ref{def:knows}(i) and Def.~\ref{def:knows}(ii) both hold for $\Gamma_1$ 
always produce
sets that are confined to $W$. Very often either $\xi(1) \not \subseteq W$ and / or $\xi(-1) \not \subseteq W$.
An example of this is given just above, in the discussion involving thermometers. As mentioned in
that discussion, for such a $\xi$,
it may be that (for example) $\xi(1)$ contains points $u' \not \in W$ such that $\Gamma_1(u') = -1$
but $\Gamma_2(u') = 1$. Now for any such $u'$, it must be that $Y(u') = -1$ (since $\DDD$ weakly infers $\Gamma_1$).
On the other hand, for any $u \in W \cap \xi(1)$, $Y(u) = 1$. Since the value of $\Gamma_2(u)$ does not change
across $\xi(1)$, this means that $Y$ and $\Gamma_2$ cannot have the same value across all
of $\xi(1)$. 

This means that that function $\xi$ could not be used to establish that $\DDD$ weakly infers $\Gamma_2$ --- and
therefore all bets are off concerning whether $\DDD$ can physically know $\Gamma_2$ over $W$.
This reflects the fact that while under the conditions of Coroll.~\ref{prop:impl}(i) $\Gamma_1$ and $\Gamma_2$ must be identical 
for all $u \in W$, they will in general differ outside of $W$, and so a device that can say both whether $\Gamma_1$ is
true or not may not be able to tell us whether $\Gamma_2$ is true or not.
\end{enumerate}

\label{ex:no_log_omni}
\end{example}

Recall from the introduction that in many epistemic logics, 
if an agent knows a proposition $\phi$ is true (more generally, that a set of propositions are true), and $\phi \Rightarrow \phi'$, 
then not only is $\phi'$ true --- but the agent knows that it is. (In particular, as discussed
above, this is true of Kripke structures.)

This property of
``(full) logical omniscience'' is a major problem with these logics, 
since logical omniscience implies for example
that if someone knows the axioms of numbers theory, then they know all
the theorems of number theory that are implied by those axioms. 
However as illustrated in Ex.~\ref{ex:no_log_omni},
physical knowledge need not obey logical omniscience.{\footnote{It is known
that so long as both the distribution axiom and the knowledge generalization rule 
hold --- which is the case in all so-called ``normal modal logics'' --- then so does (full) logical omniscience. However
neither of those need hold with physical knowledge.}}

A closely related point is that the definition of physical knowledge
does not fully agree with the colloquial meaning of the term ``knowledge''. 
It should really be viewed more of a strengthened form of inference, capturing
more of the common structure of real-world prediction, observation, memory and control,
rather than an attempt to provide an accurate entry in an English language dictionary.

For example, it is possible that a device knows that 
$\Gamma_1 \wedge \Gamma_2 = 1$ over $W$, but does not
know that $\Gamma_1 = 1$ over $W$. In particular,
physical knowledge by a device that $\Gamma_1 \wedge \Gamma_2 = 1$ 
over $W$ provides no guarantees that the device weakly infers $\Gamma_1$;
loosely speaking, the ID may not be able to correctly answer questions concerning
the value of $\Gamma_1(u)$ for $u \not \in W$, even though they can answers
concerning the value of $\Gamma_1(u) \wedge \Gamma_2(u)$ for such
$u$.{\footnote{As an example, suppose that some of the
subsets $\bar{x}$ that are images of $\xi$ extend beyond $W$, and in particular include
points $u$ at which $\Gamma_1(u) \wedge \Gamma_2(u) = -1$ while $\Gamma_1(u) = 1$.
$Y(u)$ must equal -1 for such a $u$, since the device uses $\xi$ to weakly infer $\Gamma_1 \wedge \Gamma_2$.
But this means that $\xi$ does not weakly infer $\Gamma_1$, and so does not know $\Gamma_1 = 1$
over $W$.}$^,${\footnote{This should not be surprising; if logical omniscience held, then
knowledge that $\Gamma_1 \wedge \Gamma_2 = 1$ over $W$ \emph{would}
imply knowledge that $\Gamma_1 = 1$ over $W$.}}

Nonetheless, it is worth noting that the definition of physical knowledge could
be weakened to agree with this aspect of the colloquial meaning of ``knowledge''.
One way to do that would be drop the requirement that 
the ID infer $\Gamma$ in full, including for $u \not \in W$. Under this modified definition of what it 
means for the ID to know that $\Gamma = \gamma$ over $W$,
we would still require that for all $u \in \xi(\gamma)$, if $Y(u) = 1$, then
$\Gamma(u) = \gamma$ (whether or not $u \in W$). So no matter what
$u$ is, we would require that if the device is answering the question, ``does $\Gamma(u) = \gamma$?''
and it answers `yes', then it is correct. However for all $\gamma' \ne \gamma$,
we only require that for all $u \in W \cap \xi(\gamma')$, if $Y(u) = -1$, then
$\delta_\gamma(\Gamma(u)) = Y(u)$. Under this modification, we would allow there to be $u$ outside
of $W$, and $\gamma' \ne \gamma$, where the device is answering the question, ``does $\Gamma(u) = \gamma'$?'' 
and incorrectly answers `no'.
%

\subsection{Physical knowledge that you have physical knoweldge}
\label{sec:common_knoweldge}

%

The final two rules of \emph{S5} are known as the \emph{positive introspection rule}
and the \emph{negative introspection rule}. Intuitively, they stipulate that when an agent knows something, they
know that they know it, and when they don't know something, they know that they don't know it (resp.).

Perhaps the simplest formalization of these rules occurs
in the event-based framework based on Aumann structures~\cite{fagin2004reasoning,zalta2003stanford,aumann1976agreeing,auma99,aubr95,futi91,bibr88}.
As discussed in the introduction, in this framework, in this framework events are defined as subsets of $U$. 
%
%
%
We say ``Alice knows event $E$'' if $A(u) \subseteq E$, where $A$ is Alice's
knowledge operator. So the event, ``Alice knows $E$" is just the union
of all $u$ such that  Alice knows $E$ for $U = u$, i.e., the union of all $u \in U$ such that $A(u) \subseteq E$ . (Note that
even if $u \in E$, $A(u)$ may include points $u' \not \in E$ --- no elements of such a set $A(u)$ are
contained in the event ``Alice knows $E$''.) It is immediate that if Alice knows event $E$, then Alice
knows \{Alice knows $E$\}. This is the (event-based approach version of the) positive introspection
rule. 

Physical knowledge is formulated in terms of functions and subsets $W$, not in terms of events,
so we need to extend it to consider the introspection rules. We say that ``$\DDD$ \textbf{(physically) knows} event $E \subseteq U$''
if $\DDD$ knows $\mathcal{X}_{_E} = 1$ over $E$ for some $\xi$. Next, we must define
what subset of $U$ is represented by ``the event that \{$\DDD$ knows event $E$\}'', i.e., by
 ``the event that \{$\DDD$ knows $\mathcal{X}_{_E} = 1$ over $E$ for some $\xi$\}''.
We adopt the interpretation that this set is the union of all sets $\bar{x}$ that might arise in the image of some
$\xi$ such that $\DDD$ knows $\mathcal{X}_{_E} = 1$ over $E$ for $\xi$. We write this set as
\ba
K(\DDD \; knows \; E) \;\;\; &:=& \bigcup_{\xi \; : \; \DDD \; knows \; \mathcal{X}_{_E} = 1 \; over \; E \; for \; \xi} \xi(1) \cup \xi(-1)
		\nonumber \\
\ea
(Note that in general, $K(\DDD \; knows \; E)$ can include points $u$ that lie outside of $E$.)
This allows us to translate from the event-based framework to the physical knowledge
framework: we say that ``$\DDD$ obeys positive introspection''
if for every event that $\DDD$ knows, $\DDD$ also knows the event $K(\DDD \; knows \; E)$.

\begin{corollary}
For every event that a device $\DDD$ knows, $\DDD$ also knows the event $K(\DDD \; knows \; E)$
\end{corollary}
\begin{proof}
Plugging in, ``$\DDD$ knows the event $K(\DDD \; knows \; E)$'' will be established if we can show that
\ba
\DDD \; knows \; \mathcal{X}_{_{K(\DDD \; knows \; E)}}  = 1 \; over \; K(\DDD \; knows \; E)   \nonumber
\ea
for some $\xi$.

Now by hypothesis $\DDD$ {knows event} $E$. So 
there is at least one function $\xi$ such that $\DDD$ knows $\mathcal{X}_{_E} = 1$ over $E$ for $\xi$. 
By Def.~\ref{def:knows}(i), this means that $\DDD$ weakly infers $\mathcal{X}_E$ using $\xi$.
Therefore for both $\gamma \in \B$, $u \in \xi(\gamma) \Rightarrow \delta_\gamma(\mathcal{X}_E(u)) = Y(u)$.
Moreover  for both $\gamma \in \B$, $\mathcal{X}_E(u) = \mathcal{X}_{K(\DDD \; knows \; E)}(u)$
for all $u \in \xi(\gamma)$. Therefore 
$\DDD$ weakly infers $\mathcal{X}_{_{K(\DDD \; knows \; E)}}$ using that same function $\xi$.

Next, note that $\xi(1) \cap K(\DDD \; knows \; E)$ equals $\xi(1) \cap E$. Moreover,
since $\DDD$ knows $\mathcal{X}_{_E} = 1$ over $E$ for $\xi$, 
by Def.~\ref{def:knows}(ii), $\varnothing \;\ne\; \xi(1) \cap E \;\subseteq\; Y^{-1}(1)$. 
So $\varnothing \;\ne\; \xi(1) \cap K(\DDD \; knows \; E) \;\subseteq\; Y^{-1}(1)$. Therefore the condition in
Def.~\ref{def:knows}(ii) holds for knowledge that ${\mathcal{X}}_{K(\DDD \; knows \; E)} = 1$ over $K(\DDD \; knows \; E)$ by using $\xi$ .

Similarly, $\varnothing \;\ne\; \xi(-1) \cap K(\DDD \; knows \; E) \;\subseteq\; Y^{-1}(-1)$.
So all three criteria in Def.~\ref{def:knows} are met for physical knowledge that $\mathcal{X}_{_{K(\DDD \; knows \; E)}} = 1$
over $K(\DDD \; knows \; E)$ by using $\xi$.  
\end{proof}

%
%

%
In this sense, the positive introspection rule of \emph{S5} holds for physical knowledge.

The negative introspection rule cannot hold for the event-based framework. This is because the event
``Alice does not know $E$'' cannot contain any $u$ obeying $u \in A(u)$, and so Alice
can never know that she does not know $E$. (As a result,
investigations in the event-based approach focus on negative introspection of belief rather than negative introspection 
of knowledge.)
Not surprisingly then, it is not clear how to formalize a physical knowledge version of the negative
introspection rule, since that requires defining a function over $U$ that captures
the case that the device does \emph{not} know that $u \in E$.
(N.b., that is not the same as having the device know that $u \not \in E$.)

\section{Future work}

Much more work remains to complete our understanding of inference.
Perhaps most obviously, a lot remains to be investigated concerning the relationship between 
structures like inference complexity (the ID version of Kolmogorov
complexity) and all the results in algorithmic information theory, from Chaitin's
``incompleteness theorem'' to the Halting theorem to computational complexity theory.

There is also a lot of future work to be done concerning physical knowledge. To begin, it
might be useful to extend the 
analysis of physical knowledge to include all the concepts introduced in the
analysis of inference, e.g., strong inference, covariance accuracy, and inference complexity.
Other future work on physical knowledge would be to develop Prop.\ref{prop:first_impl}, Coroll.~\ref{corr:trivial} and 
Coroll.~\ref{prop:impl} into a complete
axiomatization of physical knowledge, i.e., a set of axioms that are logically equivalent
to the definition of physical knowledge. The goal would be to parallel the kind of axiomatization
which has been done for Kripke structures (See Chap. 3 of \cite{fagin2004reasoning}.)
As a final example, it might be illuminating to construct and then investigate physical knowledge 
versions of common knowledge, of distributed knowledge, and associated results in conventional epistemic logic,
e.g., Aumann's famous proof that ``no two Bayesians can disagree''~\cite{aumann1976agreeing}.

There are also many ways to extend the concept of physical knowledge to capture
attributes of our physical world, like space and time. As an example,
%
suppose we are given a \textbf{distance function} $D(\gamma, \gamma')
: \Gamma(U) \times \Gamma(U) \rightarrow {R}^+$. We want to use that $D(., .)$ to define the distance
between what a given ID ``says'' that $\Gamma(u)$ is, and what $\Gamma(u)$ really is.
One way to do this builds on the physical knowledge formalism. For simplicity, assume $W$ refines $\Gamma$. 
We say that $\DDD$ \textbf{claims} $\gamma$ over $W$ if $\exists\;  \xi : \Gamma(U) \rightarrow \bar{X}$
such that for all $u \in \xi(\gamma) \cap W$, $Y(u) = 1$, and for all $\gamma' \ne \gamma, u \in \xi(\gamma') \cap
W$, $Y(u) = -1$. (Note that there may be more than one $\gamma$ that $\DDD$ claims over $W$.)
We define the \textbf{error} of $\DDD$ over $W$ (for $\gamma$) as the smallest $\epsilon \in \R$ such that
$\DDD$ claims $\gamma$ over $W$ and $D(\Gamma(W), \gamma) = \epsilon$.

By supposing a probability distribution over $U$ as well as a distance function, we can analyze
concepts like the expected error of the claim of a device, the variance of what a device claims, etc.
In particular, it may be possible to use an error function to extend
the analysis that led to Prop.~\ref{prop:prop6}, to investigate the possible relationship
between inference and Heisenberg's uncertainty principle. 
(As part of such an investigation,
it may be helpful to focus on the specific case where $U$ is a Hilbert space.)

Other future work is to investigate the use of inference devices in general, 
and physical knowledge in particular,
as a formalization of ``semantic information'', a concept that has been extensively
debated by people ranging from the founders of information theory and cybernetics~\cite{shannon1948mathematical}
to philosophers~\cite{stanford_encyl_semantic_info,mac1969information} to people working in
statistical physics~\cite{wolpert_banff_fqxi_semantic_info_2016,rovelli2016meaning,crutchfield1992semantics,crutchfield1992knowledge}.

Finally, despite the relation of physical knowledge with epistemic
logic, physical knowledge is designed only to capture properties of knowledge
concerning physical reality. It is not designed to capture properties of knowledge
concerning mathematical systems, e.g., predicate logic. However it may be worth investigating
its application to such systems. For example, one could take
each ``history'' in a reality to be a (perhaps infinite) string over some fixed
alphabet. $U$ might then be defined as the set of all strings
that are ``true'' under some encoding that translates a string into
axioms and associated logical implications. Then an inference device
would be a (perhaps fallible) theorem-proving algorithm, embodied
within $U$ itself. The results presented above would then concern
the relation among such theorem-proving algorithms.

\section*{Acknowledgements}

I would like to thank Philippe Binder, Artemy Kolchinsky, Brendan Tracey, Alexi Parizeau
and especially Walt Read for many useful discussions. In particular,
Prop.~\ref{prop:cov_lb}, Ex.~\ref{ex:cov_lb}, and Prop.~\ref{prop:inf_comp_entropy} are due to Walt.
This work was supported by the Santa Fe Institute, Grant No. FQXi-RFP-1622 from the FQXi foundation, 
Grant No. FQXi-RFP3-1349 from the FQXi foundation, and the Silicon Valley Foundation.

\bibliographystyle{amsplain}

\begin{thebibliography}{10}

\bibitem{stanford_encyl_semantic_info}
\emph{Stanford encyclopedia of philosophy entry on semantic information},
  $https://plato.stanford.edu/entries/information-semantic/$, 2017.

\bibitem{wiki_boolean_algebra}
\emph{Wikipedia 2017 entry on boolean algebra}, 2017,
  $https://en.wikipedia.org/wiki/Boolean\_algebra$.

\bibitem{aaronson2013philosophers}
Scott Aaronson, \emph{Why philosophers should care about computational
  complexity}, Computability: Turing, G{\"o}del, Church, and Beyond (2013),
  261--327.

\bibitem{agte05}
A.~Aguirre and M.~Tegmark, \emph{Multiple universes, cosmic coincidences, and
  other dark matters}, hep-th/0409072, 2005.

\bibitem{auma99}
R.~J. Aumann, \emph{Interactive epistemology ii: Probability}, Int. J. Game
  Theory \textbf{28} (1999), 301--314.

\bibitem{aubr95}
R.~J. Aumann and A.~Brandenburger, \emph{Epistemic conditions for nash
  equilibrium}, Econometrica \textbf{63} (1995), no.~5, 1161--1180.

\bibitem{aumann1976agreeing}
Robert~J Aumann, \emph{Agreeing to disagree}, The annals of statistics
  \textbf{4} (1976), no.~6, 1236--1239.

\bibitem{barrow2011godel}
John~D Barrow, \emph{Godel and physics}, Kurt G{\"o}del and the Foundations of
  Mathematics: Horizons of Truth (2011), 255.

\bibitem{bind08}
P.~Binder, \emph{Theories of almost everything}, Nature \textbf{455} (2008),
  884--885.

\bibitem{bibr88}
K.~Binmore and A.~Brandenburger, \emph{Common knowledge and game theory},
  ST/ICERD Discussion Paper 88/167, London School of Economics, 1988.

\bibitem{carr07}
B.~Carr (ed.), \emph{Universe or multiverse?}, Cambridge University Press,
  2007.

\bibitem{caticha2011entropic}
Ariel Caticha, \emph{Entropic dynamics, time and quantum theory}, Journal of
  Physics A: Mathematical and Theoretical \textbf{44} (2011), no.~22, 225303.

\bibitem{chaitin2004algorithmic}
Gregory~J Chaitin, \emph{Algorithmic information theory}, vol.~1, Cambridge
  University Press, 2004.

\bibitem{crutchfield1992knowledge}
James~P Crutchfield, \emph{Knowledge and meaning... chaos and complexity},
  Modeling complex phenomena (1992), 66.

\bibitem{crutchfield1992semantics}
\bysame, \emph{Semantics and thermodynamics}, Santa Fe Institute studies in the
  sciences of complexity - proceedings volume, vol.~12, Addison Wesley, 1992,
  pp.~317--317.

\bibitem{fagin2004reasoning}
Ronald Fagin, Joseph~Y Halpern, Yoram Moses, and Moshe Vardi, \emph{Reasoning
  about knowledge}, MIT press, 2004.

\bibitem{frieden2004science}
B~Roy Frieden, \emph{Science from fisher information: A unification}, Cambridge
  University Press, 2004.

\bibitem{futi91}
D.~Fudenberg and J.~Tirole, \emph{Game theory}, MIT Press, Cambridge, MA, 1991.

\bibitem{goyal2010information}
Philip Goyal, \emph{From information geometry to quantum theory}, New Journal
  of Physics \textbf{12} (2010), no.~2, 023012.

\bibitem{goyal2012information}
\bysame, \emph{Information physics---towards a new conception of physical
  reality}, Information \textbf{3} (2012), no.~4, 567--594.

\bibitem{hartle2005physics}
James~B Hartle, \emph{The physics of now}, American journal of physics
  \textbf{73} (2005), 101.

\bibitem{hopcroft2000jd}
John~E Hopcroft, Rajeev Motwani, and Ullman Rotwani, \emph{Jd: Introduction to
  automata theory, languages and computability}, 2000.

\bibitem{livi08}
M.~Li and Vitanyi P., \emph{An introduction to kolmogorov complexity and its
  applications}, Springer, 2008.

\bibitem{lloyd1990valuable}
Seth Lloyd, \emph{Valuable information}, Complexity, entropy and the physics of
  information (Wojciech~H Zurek, ed.), 1990, pp.~193--197.

\bibitem{lloyd2000ultimate}
\bysame, \emph{Ultimate physical limits to computation}, Nature \textbf{406}
  (2000), no.~6799, 1047--1054.

\bibitem{lloyd2006programming}
\bysame, \emph{Programming the universe: a quantum computer scientist takes on
  the cosmos}, Vintage, 2006.

\bibitem{lloyd2012turing}
\bysame, \emph{A turing test for free will}, Phil. Trans. R. Soc. A
  \textbf{370} (2012), no.~1971, 3597--3610.

\bibitem{mack60}
Donald~M. MacKay, \emph{On the logical indeterminacy of a free choice}, Mind,
  New Series \textbf{69} (1960), no.~273, 31--40.

\bibitem{mac1969information}
\bysame, \emph{Information, mechanism and meaning},  (1969).

\bibitem{parikh1987knowledge}
Rohit Parikh, \emph{Knowledge and the problem of logical omniscience.}, ISMIS,
  vol.~87, 1987, pp.~432--439.

\bibitem{popp88}
K.~Popper, \emph{The impossibility of self-prediction}, The Open Universe: From
  the Postscript to the Logic of Scientific Discovery, Routledge, 1988, p.~68.

\bibitem{rovelli2016meaning}
Carlo Rovelli, \emph{Meaning= information+ evolution}, arXiv preprint
  arXiv:1611.02420 (2016).

\bibitem{schmidhuber2000algorithmic}
J{\"u}rgen Schmidhuber, \emph{Algorithmic theories of everything}, arXiv
  preprint quant-ph/0011122 (2000).

\bibitem{shannon1948mathematical}
Claude~Elwood Shannon and Warren Weaver, \emph{A mathematical theory of
  communication}, 1948.

\bibitem{smol02}
L.~Smolin, \emph{The life of the cosmos}, Weidenfeld and Nicolson, 2002.

\bibitem{tegmark2008mathematical}
Max Tegmark, \emph{The mathematical universe}, Foundations of Physics
  \textbf{38} (2008), no.~2, 101--150.

\bibitem{whee90}
J.H. Wheeler, \emph{Information, physics, quantum: The search for links},
  Complexity, Entropy, and the Physics of Information (W.~H. Zurek, ed.),
  Addison-Wesley, 1990.

\bibitem{wolp01}
D.~H. Wolpert, \emph{Computational capabilities of physical systems}, Physical
  Review E \textbf{65} (2001), 016128.

\bibitem{wolp08b}
D.~H. Wolpert, \emph{Physical limits of inference}, Physica D \textbf{237}
  (2008), 1257--1281, More recent version at http://arxiv.org/abs/0708.1362.

\bibitem{wolp10}
\bysame, \emph{Inference concerning physical systems}, CiE Proceedings on
  Programs, proofs, and processes, 2010, pp.~438--447.

\bibitem{wolpert_banff_fqxi_semantic_info_2016}
David~H. Wolpert, \emph{Presentation on statistical physics definition of
  semantic information},
  $http://fqxi.org/data/documents/conferences/2016-talks/Wolpert.pdf$, 2016.

\bibitem{zalta2003stanford}
Edward~N Zalta et~al., \emph{Stanford encyclopedia of philosophy}, 2003.

\bibitem{zenil2012computable}
Hector Zenil, \emph{A computable universe: Understanding and exploring nature
  as computation}, World Scientific Publishing Co., Inc., 2012.

\bibitem{zure89b}
W.~H. Zurek, \emph{Algorithmic randomness and physical entropy}, Phys. Rev. A
  \textbf{40} (1989), 4731--4751.

\bibitem{zure89a}
\bysame, \emph{Thermodynamic cost of computation, algorithmic complexity and
  the information metric}, Nature \textbf{341} (1989), 119--124.

\bibitem{zurek1990complexity}
Wojciech~H Zurek (ed.), \emph{Complexity, entropy and the physics of
  information}, Addison-Wesley, 1990.

\bibitem{zuse1969rechnender}
Konrad Zuse, \emph{Rechnender raum (calculat1ing space)},  (1969).

\end{thebibliography}

\providecommand{\bysame}{\leavevmode\hbox to3em{\hrulefill}\thinspace}
\providecommand{\MR}{\relax\ifhmode\unskip\space\fi MR }
\providecommand{\MRhref}[2]{%
  \href{http://www.ams.org/mathscinet-getitem?mr=#1}{#2}
}
\providecommand{\href}[2]{#2}

\end{document}